\documentclass[A4paper,11pt,oneside]{article}
\usepackage[margin=2cm]{geometry}
\pdfoutput=1
\usepackage[utf8]{inputenc}

\usepackage{amsthm, amsmath, amssymb}
\usepackage{mathtools}
\usepackage{enumerate}
\usepackage{float}
\usepackage{soul}
\usepackage{subcaption}
\usepackage{xspace}
\usepackage{enumitem}
\usepackage{titling}

\usepackage{tikz}
\usetikzlibrary{arrows,shapes,trees}
\usetikzlibrary{arrows.meta}

\newcommand{\cdag}{DAG\xspace}
\newcommand{\cdags}{DAGs\xspace}
\newcommand{\colorful}{shaded\xspace}
\newcommand{\sink}{S}

\newcommand{\pZ}{\mathbb{Z}^+}
\DeclareMathOperator{\poly}{poly}
\setlist[enumerate]{topsep=4pt,itemsep=-3ex,parsep=3ex}

\newtheorem{theorem}{Theorem}
\newtheorem{corollary}[theorem]{Corollary}
\newtheorem{lemma}[theorem]{Lemma}
\newtheorem{definition}[theorem]{Definition}

\makeatletter
\@fleqnfalse
\makeatother

\title{\vspace{20pt} Red-Blue Pebbling with Multiple Processors: \\ \vspace{10pt} Time, Communication and Memory Trade-offs \vspace{70pt}}

\author{\large Toni B\"ohnlein \footnotesize \vspace{2.5pt} \\ \small Computing Systems Lab, \\ \small Huawei Zurich Research Center \vspace{1pt} \\ \footnotesize \texttt{toni.boehnlein@huawei.com}
\and \large P\'al Andr\'as Papp \footnotesize \vspace{2.5pt} \\ \small Computing Systems Lab, \\ \small Huawei Zurich Research Center \vspace{1pt} \\ \footnotesize \texttt{pal.andras.papp@huawei.com}
\and \large Albert-Jan N. Yzelman \footnotesize \vspace{2.5pt} \\ \small Computing Systems Lab, \\ \small Huawei Zurich Research Center \vspace{1pt} \\ \footnotesize \texttt{albertjan.yzelman@huawei.com}}
\date{}

\begin{document}

\begin{titlingpage}
\maketitle
\vspace{90pt}
\begin{abstract}
The well-studied red-blue pebble game models the execution of an arbitrary computational DAG by a single processor over a two-level memory hierarchy. We present a natural generalization to a multiprocessor setting
where each processor has its own limited fast memory, and all processors share unlimited slow memory. 
To our knowledge, this is the first thorough study that combines pebbling and DAG scheduling problems, capturing the computation of general workloads on multiple processors with memory constraints and communication costs. 
Our pebbling model enables us to analyze trade-offs between workload balancing, communication and memory limitations, and it captures real-world factors such as superlinear speedups due to parallelization.

Our results include upper and lower bounds on the pebbling cost, an analysis of a greedy pebbling strategy, and an extension of NP-hardness results for specific DAG classes from simpler models. 
For our main technical contribution, we show two inapproximability results that already hold for the long-standing problem of standard red-blue pebbling: (i) the optimal I/O cost cannot be approximated to any finite factor, and (ii) the optimal total cost (I/O+computation) can only be approximated to a limited constant factor, i.e., it does not allow for a polynomial-time approximation scheme. These results also carry over naturally to our multiprocessor pebbling model.
\end{abstract}
\end{titlingpage}
\setcounter{page}{2}

\pagenumbering{arabic}

\newpage

\section{Introduction}

The computational requirements of modern applications, ranging from scientific simulations to artificial intelligence, necessitate parallel data processing. 
However, developing parallel algorithms that efficiently divide (sub)tasks, manage memory, and organize communication between processors presents significant difficulties to computer scientists.
One particular bottleneck for the performance of parallel computations is due to data locality, i.e., the memory of modern hardware features a hierarchical structure, and to process data, it has to be stored in the lowest layer.
Since the layers' memory is limited in size, data movements between them (called I/O operations) are necessary and impact the execution time of computations significantly.
The phenomenon becomes especially noticeable when dealing with tasks involving basic operations on large amounts of input data, e.g., the training of neural networks. 

In this paper, we introduce a model offering improved insights into the challenges regarding 
I/O costs and limited memory in the context of parallel computing.
Our focus is directed toward the efficient execution of a specific computation by several processors (rather than the design of a parallel algorithm tailored for a particular problem).
We consider 
directed acyclic graphs (\cdags) as the model for a computation.
The nodes correspond to single operations, while the directed edges express 
that the output data of a node is required as input for another node, enforcing
precedence constraints on their order of execution.
We are concerned with devising effective and efficient schedules for a given \cdag and number of processors with limited working memory.

Hong and Kung~\cite{RBpebbling1} introduced the red-blue pebble game to study the I/O complexity of computations when executed by a \emph{single} processor with a two-level memory hierarchy. The processor has \emph{fast memory} (working memory) of limited capacity. Computing a node (and storing its output data in fast memory) requires that its predecessors' output data is stored in fast memory. To free up space in fast memory, I/O operations can transfer data to (and back from) \emph{slow memory}, which has unlimited capacity.

In this model, available data is indicated by placing a \emph{pebble} on the corresponding node.
Data stored in fast memory is represented by a red pebble, and data in slow memory by a blue pebble.
The game starts with an empty \cdag. 
A red pebble can be placed on a node (i.e., computing the node) if all its predecessors have red pebbles on them. 
A red pebble can be placed on a node if it already has a blue pebble, and vice versa (I/O operations). 
Any pebble can be removed for free. 
The goal is to place pebbles on the sinks (outputs of the computation) while minimizing the number of I/O operations.
The number of red pebbles may not exceed a given fast memory capacity at any time.

Red-blue pebbling studies trade-offs between I/O costs and memory size. In our paper, we extend this model to a setting where several processors compute in parallel, and communication is organized via shared memory; this allows us to study trade-offs between time, communication and memory size.

To familiarize ourselves with pebbling, consider the example \cdag depicted in Figure~\ref{fig:smpl_example}. 
First, assume that we pebble it using only $3$ red pebbles on a single processor. 
We may start by placing red pebbles on nodes $v_1$ and $v_2$, which then permits us to place a red pebble on node $v_3$. 
The red pebbles on $v_1$ and $v_2$ are no longer needed and can be removed. 
Observe that to pebble node $v_4$, we again need all $3$ red pebbles.
Since $v_3$ has an outgoing edge to node $v_5$, which we did not pebble yet, we place a blue pebble on $v_3$ (casting our first I/O operation), and then remove the red pebble from $v_3$. 
Next, $v_4$ (and its predecessors) can be pebbled analogously to $v_3$, leading to a configuration where we have a red pebble on $v_4$ and a blue pebble on $v_3$. 
We use another I/O operation to put a red pebble back on $v_3$, since it has a blue pebble.
Now, we can place a red pebble on node $v_5$, and then remove the red pebbles from $v_3$ and $v_4$.  
Note that we pebbled the sub-graph in the dashed box with $2$ I/O operations.
Another I/O operation places a blue pebble on $v_5$.
We pebble the subtree rooted $v_6$ in an identical way to the one rooted at $v_5$.
To finish, we bring a red pebble back to $v_5$ with an I/O step, and then place a red pebble on $v_7$.

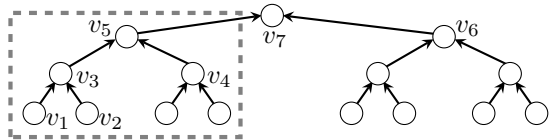
\begin{figure}[t]
    \centering
    \vspace{-7pt}
    \begin{tikzpicture}

    \begin{scope}[thick, arrows=-stealth]
    \draw (1pt,5pt) -- (8pt,15pt);
    \draw (19pt,5pt) -- (12pt,15pt);

    \draw (51pt,5pt) -- (58pt,15pt);
    \draw (69pt,5pt) -- (62pt,15pt);

    \draw (121pt,5pt) -- (128pt,15pt);
    \draw (139pt,5pt) -- (132pt,15pt);

    \draw (171pt,5pt) -- (178pt,15pt);
    \draw (189pt,5pt) -- (182pt,15pt);

    \draw (12pt,20pt) -- (30.75pt,31pt);
    \draw (62pt,20pt) -- (39pt,31pt);

    \draw (132pt,20pt) -- (151pt,31pt);
    \draw (182pt,20pt) -- (159pt,31pt);

    \draw (37pt,33pt) -- (86pt,40pt);
    \draw (152pt,33pt) -- (94pt,40pt);

    \end{scope}

    \draw[gray, dashed, ultra thick] (-9pt,-6pt) rectangle (78pt,41pt);

    \draw[black, fill=white] (0pt,3pt) circle (.9ex);
    \node[anchor=center] at (9pt,0pt) {\small $v_1$};
    \draw[black, fill=white] (20pt,3pt) circle (.9ex);
    \node[anchor=center] at (29pt,0pt) {\small $v_2$};
    \draw[black, fill=white] (10pt,18pt) circle (.9ex);
\node[anchor=center] at (20.5pt,16pt) {\small $v_3$};
    
    \draw[black, fill=white] (50pt,3pt) circle (.9ex);
    \draw[black, fill=white] (70pt,3pt) circle (.9ex);
    \draw[black, fill=white] (60pt,18pt) circle (.9ex);
    \node[anchor=center] at (70pt,16pt) {\small $v_4$};

    \draw[black, fill=white] (35pt,32pt) circle (.9ex);
    \node[anchor=center] at (25pt,35pt) {\small $v_5$};

    \draw[black, fill=white] (120pt,3pt) circle (.9ex);
    \draw[black, fill=white] (140pt,3pt) circle (.9ex);
    \draw[black, fill=white] (130pt,18pt) circle (.9ex);

    \draw[black, fill=white] (170pt,3pt) circle (.9ex);
    \draw[black, fill=white] (190pt,3pt) circle (.9ex);
    \draw[black, fill=white] (180pt,18pt) circle (.9ex);

     \draw[black, fill=white] (155pt,32pt) circle (.9ex);
     \node[anchor=center] at (164pt,35pt) {\small $v_6$};

     \draw[black, fill=white] (90pt,40pt) circle (.9ex);
     \node[anchor=center] at (91pt,31pt) {\small $v_7$};   
\end{tikzpicture}
    \vspace{-3pt}
    \caption{A simple example \cdag for pebbling.}
    \label{fig:smpl_example}
    \vspace{-5pt}
\end{figure}

In contrast, consider pebbling the \cdag with multiple processors for an intuitive description of our model. 
We use two processors $p_1$ and $p_2$ each having their own set of $3$ red pebbles. To place a red pebble on a node, its predecessors require red pebbles of the \emph{same} processor. We pebble the left and the right subtree of node $v_7$ with red pebbles of $p_1$ and $p_2$, respectively, using the strategy described above for a single processor. All steps are executed in parallel, reducing both the compute and I/O operations by a factor $2$. The result is a configuration with a red pebble of $p_1$ and $p_2$ on nodes $v_5$ and $v_6$, respectively.

Recall that with one processor, we needed two I/O operations to convert the red pebble on $v_5$ to blue and then back to red.
In contrast, with two processors we need I/O operations to ensure that we have red pebbles of the same processor on $v_5$ and $v_6$.
The required communication between the processors is done via the shared memory (blue pebbles):
We replace the red pebble of $p_1$ on $v_5$ by a blue pebble, and then the blue pebble by a red pebble of $p_2$. To finish the computation, we place a red pebble of $p_2$ on $v_7$.

\vspace{7pt}
\noindent
\textbf{Our Contribution. }
We present the multiprocessor red-blue pebbling game (MPP), which is a natural generalization of red-blue pebbling. MPP essentially combines the areas of DAG pebbling problems (single-processor computation with limited memory) and DAG scheduling problems (multiprocessor computation, but with unlimited memory). The result is a simple yet expressive model that captures computational costs, I/O costs incurred by memory limitations, and inter-processor communication costs. 
Unlike earlier attempts to generalize red-blue pebbling, MPP naturally combines these aspects into a single cost function, allowing for a convenient study of the trade-offs between these factors. 
These trade-offs have been analyzed in detail for many concrete computations before (e.g., matrix multiplication~\cite{RBpebbling6, RBpebbling9,multiBSP2}), but not for arbitrary computational \cdags.

We first discuss the basic properties of MPP, including simple bounds on the pebbling cost, NP-hardness for specific DAG classes, and the analysis of a simple greedy approach. We also analyze how adding more processors (with the same, or reduced amount of fast memory) can affect the pebbling problem, i.e., how the optimal pebbling cost can change in the same DAG. 

Then as our main technical result, we present two hardness results for approximating the optimal strategy in red-blue pebbling problems. Both of these proofs already apply to single-processor red-blue pebbling, thus presenting novel results for a problem that has been studied for several decades. Firstly, we show that the minimal number of I/O steps in a pebbling cannot be approximated in polynomial time to any finite multiplicative factor, or any additive $n^{1-\varepsilon}$ term (for any $\varepsilon>0$), unless P=NP. Secondly, we show that the minimal total cost (including both I/O and computation steps) can only be approximated to a limited factor, i.e., there is a $\delta>0$ such that the optimum cannot be approximated to a $(1+\delta)$ factor in polynomial time, unless P=NP; in other words, the problem is APX-hard, allowing no PTAS. Both of these results carry over naturally to our multiprocessor pebbling model.

\section{Related Work}
Pebble games are used to capture various aspects of computing. For instance, the standard (black) pebble game \cite{cook1973observation, paterson1970comparative} models general time-memory trade-offs, with results on achievable trade-offs \cite{hopcroft1977time, lengauer1982asymptotically, reischuk1980improved} and the complexity of computing pebbling strategies \cite{gilbert1979pebbling, demaine2017inapproximability, blackNPhard, austrin2012inapproximability}. 

Hong and Kung~\cite{RBpebbling1} introduce the red-blue pebble game to study I/O complexity, and derive lower bounds based on a \cdag partition technique. 
It is applied to specific computations in \cite{elango2015characterizing, ranjan2012upper}. 
Further bounds on I/O costs were derived with different methods in \cite{jain2020spectral, scott2015matrix, RBpebbling5}.
Other works~\cite{savage1995extending,savage2008unified,elango2014characterizing} extend the red-blue pebble game to 
single and multiple processors over memory hierarchies. 
Kwasniewski et al.~\cite{RBpebbling9, RBpebbling6} refine the technique of \cite{RBpebbling1} and derive improved I/O lower bounds for special \cdags. The bounds are extended to a multiprocessor settings where the workload is perfectly balanced. 
In contrast to MPP, 
these works do not consider 
trade-offs between computation and I/O.

As for the complexity of computing optimal pebbling strategies, Demaine and Liu~\cite{RBpebbling2} show that standard red-blue pebbling is PSPACE-complete, and propose variations that are shown NP-complete. 
NP-hardness for red-blue pebbling with computation costs is also shown in~\cite{RBpebbling8}. The work of~\cite{RBpebbling3} shows an inapproximability to a $(2-\varepsilon)$ factor, assuming the unique games conjecture. The work of~\cite{RBpebbling7} presents a bi-criteria approximation for the optimal I/O cost to a $\poly \log(n)$ factor, provided that it is allowed to violate the memory bound by a $\poly \log(n)$ factor.

Naturally, there are also several more realistic models that capture both parallelization and memory limitations, but these exhibit key differences from MPP. The work of \cite{blelloch2008provably, chowdhury2008cache} studies the cost components (computation, different kinds of I/O) separately, mostly assuming that I/O steps cannot be parallelized. 
Other works assume simpler models regarding computation costs: they are not considered at all in \cite{bilardi2018lower}, and are handled with artificial balance constraints in \cite{solomonik2017trade, RBpebbling5}.
Moreover, these models are primarily used for studying specific computations (e.g., Matrix multiplication, FFT) rather than arbitrary \cdags. 

Valiant~\cite{BSPintro} introduces the bulk-synchronous parallel (BSP) computing model with explicit synchronisation, which was extended later~\cite{tiskin1998bulk,mccoll24memory} to include data locality and communication via shared memory. The model was further generalized to capture multi-processor/core architectures with multi-level memory~\cite{multiBSP2}. These models describe practical compute systems more closely, and also match our MPP model. There are also theoretical studies on \cdag scheduling in BSP \cite{BSPscheduling}, but without memory restrictions.

\section{Model Definition} \label{sec:model}

Our model for a computation is a directed acyclic graph (\cdag) $G = (V,E)$.
We use the notation $n = |V|$ for the number of nodes.
The in/out-degree of a node is the number of its incoming/outgoing edges.
We call nodes with in/out-degree $0$ the source/sink nodes.
Let $\Delta_{in}$ denote the highest in-degree in our \cdag, i.e.,\ the highest number of inputs to an operation in our computation; several previous works on pebbling assume that $\Delta_{in}$ is small, e.g., a constant \cite{RBpebbling2, RBpebbling3}. Moreover, we denote the set of positive integers by $\pZ$, and use the notation $[a]=\{1, ..., a\}$, for $a \in \pZ$. Our goal is to execute a computational \cdag on $k$ processors, each having limited fast memory of size $r$, for parameters $k,r \in \pZ$.

\subsection{Single Processor Red-Blue Pebbling}

Before introducing our parallel model, we briefly recap the single-processor red-blue pebble game (SPP) \cite{RBpebbling1}, where we have a single processor with fast memory of size $r$ and unlimited slow memory. Nodes with red pebbles on them (denoted $R \subseteq V$) and nodes with blue pebbles on them (denoted $B \subseteq V$) correspond to the output data that is currently saved in fast and slow memory, respectively. 
In order to pebble a \cdag, the following rules can be applied:
\begin{enumerate}[label=(R\arabic*-S), leftmargin=3.5em, itemsep=-12pt]
	\item Place a red pebble on a node that has a blue pebble, \label{rule:single-redbyblue}
	\item Place a blue pebble on a node that has a red pebble, \label{rule:single-bluebyred}
	\item Place a red pebble on a node if all its predecessors have red pebbles on them, \label{rule:single-compute}
	\item Remove a (red or blue) pebble. \label{rule:single-remove}
\end{enumerate}

The game starts with an empty \cdag; rule \ref{rule:single-compute} allows us to place red pebbles on source nodes. A pebbling strategy is a sequence of the rules which places pebbles on the sink nodes in the end, while the number of red pebbles does not exceed $r$ at any step.
The I/O costs are the total number I/O operations, i.e., applications of rules \ref{rule:single-redbyblue} or \ref{rule:single-bluebyred}.
The goal is to pebble a \cdag with minimum I/O costs. Most works on SPP analyze lower bounds on I/O cost for a given \cdag and memory limit $r$. 

Additional variations of SPP have also been proposed, aiming to simplify SPP, make it more realistic, or resolve the fact that this base SPP variant above is not even in NP, but rather PSPACE-complete. This is because the optimal pebbling sequence can be super-polynomially long in extreme cases, since repeatedly deleting and recomputing the same node incurs no cost. Notable SPP variants include:
\vspace{-5pt}
\begin{itemize}[leftmargin=1.5em, itemsep=-3pt]
    \item \emph{One-shot SPP:} (R3-S) can be applied only once for each node~\cite{RBpebbling7},
    \item \emph{No-deletion SPP:} (R4-S) is not allowed~\cite{RBpebbling2},
    \item \emph{SPP with computation costs:} (R3-S) also incurs a small cost of $\varepsilon$~\cite{RBpebbling3, RBpebbling8}.
\end{itemize}
\vspace{-4pt}
The first two variants feature somewhat artificial restrictions to prohibit the deletion and then recomputation of a node entirely.
On the other hand, SPP with computation costs is more realistic, discouraging this in a natural way by ensuring that computation steps also incur some cost, as in practice. This last SPP variant is also the closest one our multiprocessor pebbling model.

We note that besides this, there are also smaller aspects of the definition that vary over different previous works: for instance, some assume that source nodes begin with a blue pebble, or that sink nodes specifically require a blue pebble in the end. For most proof constructions, these model variants can be reduced to each other with some simple tricks; we refer to \cite{RBpebbling3} for a summary. Similarly, some works assume in rule \ref{rule:single-redbyblue} that the new red pebble replaces the blue pebble, and in rule \ref{rule:single-bluebyred} vice versa~\cite{RBpebbling3, RBpebbling5}; our SPP-related claims also carry over easily to this variant.

\subsection{Multiple Processor Red-Blue Pebbling}

We introduce the multiprocessor red-blue pebble game (MPP) which extends the red-blue pebble game to a setting where $k$ processors compute a \cdag $G = (V,E)$ in parallel. 
Data stored in a processor's fast memory is represented by a red pebble of its \emph{shade}, i.e., there are red pebbles of $k$ different shades.
The number of red pebbles of each shade is limited by $r$.
The processors share unlimited slow memory that is used to
(i) store data that cannot be kept in fast memory, and 
(ii) to communicate data between processors.
Data available in the slow memory is represented by blue pebbles. Note that we assume that several pebbles of different shade/color can be placed on a node at the same time.

In our parallel version of the game, the rules allow multiple processors to simultaneously either transfer data between fast and slow memory, or compute nodes. 
More formally, define set $R^j \subseteq V$ as the set of \emph{red pebbles} of \emph{shade $j$}, for $j \in [k]$, and
set $B \subseteq V$ as the set of \emph{blue pebbles}.
For $m \in \pZ$ such that $m \leq k$, we call an injective function $f_m: [m] \to [k]$ a \emph{\colorful selection}, and
extend the transition rules as follows: 
\begin{enumerate}[label=(R\arabic*-M), align=right, leftmargin=3.5em, itemsep=-10pt]
\item \label{rule:multi-redbyblue} For a \colorful selection $f_m$ and vertices $v_1,v_2, \ldots, v_m$ such that $v_i \in R^{f_{m}(i)}$, add $v_i$ to $B$, for $i \in [m]$,  
\item For a \colorful selection $f_{m}$ and vertices $\{v_1,v_2, \ldots, v_m\} \subseteq B$, add $v_i$ to $R^{f_{m}(i)}$, for $i \in [m]$,
 \label{rule:multi-bluebyred}
\item For a \colorful selection $f_{m}$ and vertices $v_1,v_2, \ldots, v_m$ such that for all predecessors $u$ of $v_i$ we have $u \in R^{f_{m}(i)}$, for $i\in [m]$, add $v_i$ to $R^{f_{m}(i)}$, for $i\in [m]$, 
\label{rule:multi-compute}
\item Remove a (red or blue) pebble. \label{rule:multi-remove}
\end{enumerate}
A parallel version of \ref{rule:multi-remove} could also be defined; we use this simpler version since this step incurs no cost in our model. 
Note that any single rule can place at most $k$ pebbles, and that processors can be idle.

A \emph{configuration} is a $(k\!+\!1)$-tuple $C_i = (R^1_i,R^2_i, \ldots, R^k_i, B_i) \subseteq V^{k+1}$
defining red and blue pebbles placed on the \cdag.
A configuration $C_i$ is \emph{valid} if $|R^j_i| \leq r$, for $j \in [k]$, i.e., it respects the fast memory size. 
An initial configuration $C_0$ sets $R^j_0 = \emptyset$, for $j \in [k]$, and $B_0 = \emptyset$.
Let $\sink \subseteq V$ be the sink nodes.
We say configuration $C_i$ is terminal, if $\sink \subseteq B_i \bigcup_{j=1}^k R_i^j$ holds.
A \emph{pebbling strategy} $(C_0, C_1, \ldots, C_T)$, for $T \in \pZ$, is a sequence of valid configurations such that (i) $C_{i}$ is obtained by applying a transition rule to $C_{i-1}$, for $i \in [T]$, and (ii) $C_0$ and $C_T$ are initial and terminal, respectively. 

The pebbling strategy is represented equivalently by a sequence of transition rules
$(t_1,t_2, \ldots, t_{T})$, where $t_i \in \{$\ref{rule:multi-redbyblue},  \ref{rule:multi-bluebyred}, \ref{rule:multi-compute}, \ref{rule:multi-remove}$\}$ such that $C_i$ is the result of applying $t_i$ to $C_{i-1}$, for $i \in [T]$. 
To assign costs to a pebbling strategy, we assign costs to each rule.
Let $g \in \pZ$ be a parameter specifying the cost of an I/O step. We define the cost function for the transition rules as follows:
\begin{itemize}[itemsep=-3pt, topsep=1pt]
    \item $c(t_i) = g \text{ if $t_i =$ \ref{rule:multi-redbyblue} or \ref{rule:multi-bluebyred}}$,
   \item $c(t_i) = 1 \text{ if $t_i =$ \ref{rule:multi-compute}}$,
    \item $c(t_i) = 0 \text{ if $t_i =$ \ref{rule:multi-remove}}$.
\end{itemize}

Then, the cost of a strategy $C(t_1,t_2, \ldots, t_T) = \sum_{i=1}^{T} c(t_i)$ is the total costs of its rules.
Given a \cdag $G$ and parameters $k,r,g \in \pZ$, the
goal of MPP is to find a minimum cost pebbling strategy.
We denote the cost of the optimum pebbling strategy by \texttt{OPT}.

\subsection{Model Discussion}

We first note that as in SPP, the rules of MPP allow for \emph{recomputation}. Indeed, when I/O is expensive, it can sometimes be beneficial to compute the same node more than once; see Section~\ref{sec:basics} for an example.

Furthermore, the transition rules assume that the processors compute and access memory \emph{synchronously}.
This simplifies the analysis greatly, enabling us to formulate the cost function as a linear term of I/O and compute costs.
However, in some practical settings, one processor may be computing while another one is accessing memory.
This could be modelled by allowing each processor in a step to execute one of the SPP rules independently;
however, assigning costs to a pebbling strategy then becomes an intricate matter.
We expect that most of the general reasoning about pebbling strategies also carries over to such an asynchronous setting; it has been shown that the improvements from a non-synchronous schedule are limited to a factor $2$ \cite{BSPscheduling}. Synchronization of communication is also natural in some hardware architectures, and several parallel programming models, like BSP~\cite{BSPintro}, also feature synchronization steps.

While parallel computing models typically study trade-offs between computation time and communication, SPP studies trade-offs between I/O costs and memory size. Combining these in MPP allows us to study the three-fold trade-off between computation time, communication, and memory size. The I/O steps in MPP can happen either due to (i) communicating data between processors, or (ii) saving data to slow memory to free up space in fast memory.

We note that several models in previous works are closely related to MPP; we discuss these in detail in Appendix~\ref{app:models}. For instance, \cite{RBpebbling5} also outlines a generalization of SPP to multiple processors; however, in contrast to MPP, their work captures computation costs via an artificial balance constraint, which imposes heavy limitations on the model (see Appendix \ref{sec:constrained_MPP}). We also note that with $r _{\!}= _{\!} \infty$ and minor adjustments, MPP also becomes equivalent to DAG scheduling in the BSP model \cite{BSPscheduling}. This shows that with small variations, MPP is indeed a generalization of both SPP and DAG scheduling problems.

\section{Fundamental Properties of MPP}\label{sec:basics}

\subparagraph*{Straightforward bounds.} Similarly to SPP, if $r \leq \Delta_{in}$, then there can be no valid pebbling strategy, since a node of in-degree $\Delta_{in}$ requires $(\Delta_{in}+1)$ red pebbles of the same shade to be computed: $\Delta_{in}$ on its in-neighbors, and one more on the node itself. Thus we always implicitly assume $r \geq \Delta_{in}+1$.

However, if $r \geq \Delta_{in}+1$, there is indeed always a valid pebbling. Consider the nodes in any topological order. We can always select any processor $p$ to compute the next node $v$, load all the (already computed) in-neighbors of $v$ from slow memory to $p$ (at a cost of at most $\Delta_{in\!} \cdot _{\!} g$), compute the node on $p$ (at a cost of $1$), save the value of $v$ to slow memory (at a cost of $g$), and then remove the red pebbles from $v$ and its in-neighbors (for free). This strategy incurs a cost of at most $(\Delta_{in}+1) \cdot g + 1$ for each node.

On the other hand, since each \ref{rule:multi-compute} step can compute at most $k$ nodes, the number of compute steps is at least $\frac{n}{k}$. This shows the following simple bounds for the optimal cost.
\begin{lemma} \label{lem:trivial_bounds}
For any instance of MPP, we have $\frac{n}{k} \leq \texttt{OPT} \leq (g \cdot (\Delta_{in} + 1) + 1) \cdot n$.
\end{lemma}

\subparagraph{A simple example gadget.} 

We briefly discuss a simple example, the \emph{zipper gadget}, which highlights many vital aspects of pebbling problems, and different variants of it are used throughout our proofs. We only describe the gadget and these properties here on a high level, leaving the details to the corresponding proofs in the appendices. The gadget, shown in Figure \ref{fig:gadget}, consists of two input groups $S_1$, $S_2$ of $d$ nodes each, and a \emph{main chain} of $n_0$ nodes, with the input groups having edges to the main chain nodes in an alternating fashion. We typically have $n_0=n-O(1)$ to make $S_1$ and $S_2$ asymptotically irrelevant.
\begin{itemize}[itemsep=-2pt, topsep=4pt, leftmargin=2em]
 \item The gadget was used in \cite{RBpebbling3} to analyze trade-offs in SPP. E.g.\ for $r\!=\!2d\!+\!2$, we can always keep red pebbles on both $S_1$ and $S_2$, and use the last $2$ red pebbles to compute the main chain without any I/O steps. However, for $r\!=\!d\!+\!2$, we still need $2$ red pebbles to compute along the main chain, so we repeatedly need to move the other $d$ red pebbles between $S_1$ and $S_2$, yielding a much higher cost.
 \item This example with $r\!=\!d\!+\!2$ also highlights the role of recomputation: we repeatedly need to move $d$ red pebbles between $S_1$ and $S_2$, and doing this with I/O steps (repeatedly loading from slow memory) incurs a cost of $d \cdot g$ for each main chain node. However, since the nodes in $S_1$ and $S_2$ are sources, we can also simply compute them again with \ref{rule:multi-compute} steps anytime, at a much lower cost of $d$ per main chain node. To ensure that such a recomputation is always suboptimal, we can also attach a chain of length $2g$ in front of each node $u_i$: then recomputing $u_i$ from a source node requires $2g+1$ compute steps, whereas saving $u_i$ to slow memory and loading it back later costs at most $2g$.
 \item Considering the gadget in MPP, if $k\!=\!1$ and $r\!=\! d\! +\! 2$, we again need to keep moving the $d$ red pebbles between $S_1$ and $S_2$, which incurs a cost of $(d _{\!} \cdot _{\!} g + 1)$ per main chain node. However, with $k=2$ processors and $r\!=\! d\! + \!2$, we can keep $S_1$ and $S_2$ in the fast memory of different processors, compute the main chain alternatingly, and only communicate these main chain nodes. This incurs a cost of $(2 _{\!} \cdot _{\!} g + 1)$ per main chain node, thus describing a superlinear speedup for larger $d$ values.
\end{itemize}

\begin{figure}
    \centering
    \vspace{-4pt}
    \resizebox{0.4\textwidth}{!}{
    \begin{tikzpicture}

    \begin{scope}[thick, arrows=-stealth]
    \draw (50pt,5pt) -- (72pt,-3pt);
    \draw (75pt,-5pt) -- (97pt,3pt);
    \draw (100pt,5pt) -- (122pt,-3pt);
    \draw (125pt,-5pt) -- (147pt,3pt);
    \draw (150pt,5pt) -- (172pt,-3pt);
    \draw (40pt,27pt) -- (49pt,9pt);
    \draw (90pt,27pt) -- (99pt,9pt);
    \draw (140pt,27pt) -- (149pt,9pt);
    \draw (65pt,-27pt) -- (74pt,-9pt);
    \draw (115pt,-27pt) -- (124pt,-9pt);
    \draw (165pt,-27pt) -- (174pt,-9pt);
    \end{scope}
    \begin{scope}[thick]
    \draw (175pt,-5pt) -- (188pt,0pt);
    \draw (20pt,10pt) -- (35pt,27pt);
    \draw (20pt,20pt) -- (35pt,27pt);
    \draw (20pt,45pt) -- (35pt,27pt);
    \draw (35pt,27pt) -- (188pt,27pt);
    \draw (20pt,-10pt) -- (35pt,-27pt);
    \draw (20pt,-35pt) -- (35pt,-27pt);
    \draw (20pt,-45pt) -- (35pt,-27pt);
    \draw (35pt,-27pt) -- (188pt,-27pt);
    \end{scope}

    \draw[black, fill=white] (50pt,5pt) circle (0.9ex);
    \draw[black, fill=white] (75pt,-5pt) circle (0.9ex);
    \draw[black, fill=white] (100pt,5pt) circle (0.9ex);
    \draw[black, fill=white] (125pt,-5pt) circle (0.9ex);
    \draw[black, fill=white] (150pt,5pt) circle (0.9ex);
    \draw[black, fill=white] (175pt,-5pt) circle (0.9ex);

    \draw[black, fill=white] (20pt,10pt) circle (0.9ex);
    \draw[black, fill=white] (20pt,20pt) circle (0.9ex);
    \node[anchor=center] at (20pt,30pt) {$...$};
    \draw[black, fill=white] (20pt,45pt) circle (0.9ex);

    \draw[black, fill=white] (20pt,-45pt) circle (0.9ex);
    \draw[black, fill=white] (20pt,-35pt) circle (0.9ex);
    \node[anchor=center] at (20pt,-25pt) {$...$};
    \draw[black, fill=white] (20pt,-10pt) circle (0.9ex);

    \node[anchor=center] at (53pt,-4pt) {$v_1$};
    \node[anchor=center] at (78pt,3.5pt) {$v_2$};
    \node[anchor=center] at (103pt,-4pt) {$v_3$};
    \node[anchor=center] at (128pt,3.5pt) {$v_4$};
    \node[anchor=center] at (153pt,-4pt) {$v_5$};
    \node[anchor=center] at (178pt,3.5pt) {$v_6$};

    \node[anchor=center] at (11pt,-51pt) {$u_1$};
    \node[anchor=center] at (8pt,-37pt) {$u_2$};
    \node[anchor=center] at (8pt,-12pt) {$u_d$};

    \node[anchor=center] at (4pt,7pt) {$u_{d_{\!}+_{\!}1}$};
    \node[anchor=center] at (4pt,17pt) {$u_{d_{\!}+_{\!}2}$};
    \node[anchor=center] at (6pt,42pt) {$u_{2d}$};
    
    \begin{scope}[gray, thick, arrows=-stealth]
    \draw (2pt,-45pt) -- (16pt,-45pt);
    \draw (-28pt,-45pt) -- (-11pt,-45pt);
    \end{scope}
    \begin{scope}[gray]
    \node[anchor=center] at (-5pt,-47pt) {$...$};
    \draw[fill=white] (-28pt,-45pt) circle (0.9ex);
    \end{scope}

    \draw[white, fill=white] (200pt,0pt) circle (0.9ex);

\end{tikzpicture}}
    \vspace{-5pt}
    \caption{Zipper gadget consisting of $2$ input groups $S_{1\!} =_{\!}\{u_1, u_2, \ldots, u_d \}$ and $S_{2\!} =_{\!} \{u_{d+1}, \ldots, u_{2d} \}$, and a main chain $v_1, \ldots, v_{n_0}$. The edges going from the input groups are combined into a single arrow for simplicity. The extension to discourage recomputation is only illustrated for $u_1$ in gray.}
    \label{fig:gadget}
\end{figure}
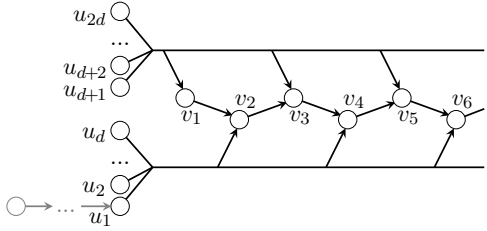

\subparagraph*{NP-Hardness.}

Regarding the complexity of MPP, it is not surprising that the problem is NP-hard, since it generalizes SPP. However, MPP is already NP-hard on rather simple subclasses of DAGs.

\renewcommand*{\proofname}{Proof sketch}

\begin{lemma} \label{lem:NPhard}
MPP is already NP-hard on the following subclasses:
\begin{itemize}[itemsep=-1pt, topsep=3pt]
 \item 2-layer DAGs (where the longest path has length $1$),
 \item in-trees (where every out-degree is at most $1$).
\end{itemize}
\end{lemma}
The lemma follows from \cite{BSPscheduling}, where the same results are established for BSP scheduling; however, it requires some further work to adapt these proof constructions to MPP.

\subparagraph*{A Greedy Algorithm.}

It is also natural to wonder if we can obtain good solutions e.g.\ by following a greedy pebbling strategy. We can derive a naive upper bound on this approach, using the long known results that in \cdag scheduling model memory limits or communication costs, any non-idle greedy strategy is a $2$-approximation of the optimum \cite{greedy2approx}. Employing such a greedy approach for the computations, and assuming the worst-case strategy discussed for Lemma \ref{lem:trivial_bounds}, we get the following bound.

\begin{lemma} \label{lem:greedy_upper}
Any pebbling where the compute steps follow a non-idle greedy schedule gives a $2 \cdot (g \cdot (\Delta_{in} + 1) + 1)$-factor approximation of the optimum.
\end{lemma}

Unsurprisingly, greedy strategies can also return rather bad solutions. In general, it is even non-trivial to precisely define a greedy strategy for MPP, since pebbling consist of various aspects. Here we make observations for a general class of greedy strategies: we only assume that processor $p$ always picks the next node to compute as the yet uncomputed node with the largest number (or largest fraction) of in-neighbors having a red pebble of $p$. We show a lower bound for any such greedy strategy, regardless of how compute steps are parallelized, how ties are broken, or how I/O steps are applied to compute the chosen node.

\begin{lemma} \label{lem:greedy_lower}
There exist DAGs where any such greedy pebbling algorithm is worse than the optimum
\begin{itemize}[itemsep=2pt, topsep=3pt]
 \item by a $\, \frac{1}{5} \! \cdot \! \Delta_{in\!} - _{\!} 1\, $ factor asymptotically (for any $\Delta_{in\!} = _{\!} O(1)$),
 \item by a $\, \frac{2}{3} \! \cdot \! g _{\!} + _{\!} 1 \,$ factor asymptotically (for any $g \geq 2$).
\end{itemize}
\end{lemma}

\subparagraph*{Lower Bounds on Pebbling Costs.}
We now discuss how we can apply lower bounds on the number of required I/O steps from SPP to bound the optimum cost in MPP.
The key insight is that a pebbling strategy for MPP can be implemented using a single processor (SPP) with fast memory of size $r \cdot k$. Specifically, each parallel rule can be simulated using $k$ sequential rules.

\begin{lemma} \label{lem:k_to_1}
Let $G$ be a \cdag such that an SPP pebbling strategy with fast memory of size $k \cdot r$ 
requires at least $L$ steps of I/O, for some $L \in \pZ$. 
Then, an MPP pebbling strategy with $k$ processors, each having fast memory of size $r$,
requires at least $L / k$ steps of I/O.
\end{lemma}


\begin{corollary}\label{coro:lb_pebbling_1}
Let $G$ be a \cdag such that an SPP pebbling strategy with fast memory of size $k \cdot r$ 
requires at least $L$ steps of I/O, for some $L \in \pZ$. 
Then, an MPP strategy with $k$ processors and fast memory of size $r$ each
has cost of at least $g \cdot L / k + {n} / {k}$.
\end{corollary}

It follows that we can translate lower bounds for a single processor 
to $k$ processors. Indeed, many previous works \cite{RBpebbling1,RBpebbling6} derive I/O lower bounds for specific computations in SPP, which are obtained utilizing a special \cdag partition. 
For example, Hong and Kung~\cite{RBpebbling1} derive a lower bound of $\frac{n \log n}{ \log(r k)}$ on the number of required I/Os in SPP (with fast memory size $r \cdot k$) for the $n$-point FFT \cdag. The bound translates to a lower bound of $ \frac{n}{k} \cdot (g \cdot \frac{\log n}{\log(r k)} + 1)$ on the cost in MPP for the same \cdag. 
Another well-studied computation is matrix-matrix multiplication. Kwasniewski et al.~\cite{RBpebbling6} improve the technique of \cite{RBpebbling1} and derive a lower bound of $\frac{2n^3}{\sqrt{rk}} + n^2$ in the single processor case, resulting in a lower bound of 
$\frac{n}{k} \cdot (g \cdot (\frac{2n^2}{\sqrt{rk}} + n) + 1)$ on the costs of matrix-matrix multiplication in MPP.

Finally, we show that there are instances where this lower bound is essentially tight.

\begin{lemma} \label{lem:IO_lower_strict}
For any $n$, there is a \cdag construction with $\texttt{OPT} \leq g \cdot L / k + {n}/{k} + O(1)$ in MPP.
\end{lemma}

\renewcommand*{\proofname}{Proof sketch}

\section{The Impact of More Processors: Trade-offs Between $k$, $r$ and \texttt{OPT}} \label{sec:tradeoff}

We next analyze how more available processors can affect the optimal pebbling strategy for a DAG, which captures fundamental trade-offs for parallel computing. For convenience, we compare the simplest case of $1$ processor to $k$ processors, but our proofs are easy to carry over to any $k$-fold increase in the number of processors. We use $\texttt{OPT}^{(k)}$ as a short notation for the optimal pebbling cost with $k$ processors. 

There are two natural ways to do this comparison. Let $r_0$ denote the amount of fast memory in the $1$-processor case. One option is to compare this to a setting with $k$ processors and $r:=\frac{r_0}{k}$ fast memory on each; we call this the \emph{fair} comparison, as the total fast memory over all processors remains unchanged. Another option is to compare to a setting with $k$ processors and simply $r:=r_0$ for each, i.e., all processors having the same fast memory $r_0$ as before. This \emph{practical} comparison is more relevant for applications, where computations are often parallelized by simply employing more processors of the same kind.

\subparagraph*{``Fair'' Comparison.}

We first compare the case of $1$ processor with $r _{\!} = _{\!} r_0$, to the case of $k$ processors with $r _{\!} = _{\!} \frac{r_0}{k}$. This is an interesting setting: on the one hand, the $k$ processors allow for parallelization of computations and I/O, but on the other hand, processors have less fast memory, possibly resulting in further I/O steps to save and reload data. We first show that the optimum cost can decrease by a factor $k$ at most; intuitively, this is because in the fair case, any pebbling strategy with $k$ processors can be transformed into a $1$-processor schedule with at most $k$ times larger cost. The bound is tight, as can be seen in e.g.\ a DAG with $k$ independent chains of length $\frac{n}{k}$.

\renewcommand*{\proofname}{Proof}

\begin{lemma} \label{lem:factor_k}
In the fair case, we have $\frac{\texttt{OPT}^{(k)}}{\texttt{OPT}^{(1)}} \geq \frac{1}{k}$, and there are DAGs such that $\frac{\texttt{OPT}^{(k)}}{\texttt{OPT}^{(1)}} = \frac{1}{k}$.
\end{lemma}

\noindent With the same fast memory scattered over $k$ processors, the optimum can also increase notably.

\begin{lemma} \label{lem:fair_increase}
In the fair case, there is a construction with 
\[ \frac{\texttt{OPT}^{(k)}}{\texttt{OPT}^{(1)}} \geq \frac{k-1}{k} \cdot g \cdot (\Delta_{in}-1) + 1 - o(1) \, . \]
\end{lemma}

Recall that $\texttt{OPT}^{(1)} \geq n$ and $\texttt{OPT}^{(k)} \leq (g \cdot (\Delta_{in} + 1) + 1) \cdot n$, so this bound is essentially tight for large $k$ and $\Delta_{in}$ values. Another construction shows that the optimum can also be non-monotonic in $k$.

\begin{lemma} \label{lem:nonmonotone}
In the fair case, we can have e.g.\ $\texttt{OPT}^{(2)\!} \leq \texttt{OPT}^{(1)}$ and $\texttt{OPT}^{(2)\!} \leq \texttt{OPT}^{(4)}$.
\end{lemma}

\subparagraph*{``Practical'' Comparison.}

We now compare MPP with $1$ and with $k$ processors, with the processor(s) having the same $r=r_0$ in both cases. Here the larger $k$ value only comes with advantages, so the optimum can never increase; on the other hand, it can easily decrease as before, e.g., for $k$ independent chains. What makes this setting more interesting is that the optimum can decrease by a factor larger than $k$; such a \emph{superlinear speedup} is a well-known (and highly desired) phenomenon in real-world systems. To our knowledge, MPP is the first DAG scheduling or pebbling model that naturally captures this behavior.

\begin{lemma} \label{lem:suplin}
In the practical case, for any $\varepsilon>0$, we can have $\frac{\texttt{OPT}^{(1)}}{\texttt{OPT}^{(2)}} \geq \frac{\Delta_{in}-1}{2} - \varepsilon$.
\end{lemma}

The proof idea has already been outlined in Section \ref{sec:basics} with the zipper gadget. With the appropriate $\Delta_{in}$, the lemma can achieve a speedup of any constant factor, already for $k=2$.

\subparagraph*{The number of I/O steps.}

We also briefly study the number of I/O steps separately: let $\texttt{OPT}_{I/O}\,\!^{(k)}$ denote the number of steps \ref{rule:multi-redbyblue} and \ref{rule:multi-bluebyred} in the optimal pebbling with $k$ processors. Our observations here hold for both the fair and the practical case.

Since $k=1$ requires no communication, whereas $k=2$ might, it is easy to construct a DAG where $\texttt{OPT}_{I/O}\,\!^{(1)}=0$, but $\texttt{OPT}_{I/O}\,\!^{(2)}=\Theta(n)$. More surprisingly, we can also have a similar decrease in I/O, i.e., a DAG with $\texttt{OPT}_{I/O}\,\!^{(1)}=\Theta(n)$, but $\texttt{OPT}_{I/O}\,\!^{(2)}=0$. Intuitively, this can happen when the computations can only be distributed in a very imbalanced way, so when we have $2$ processors, it becomes more beneficial to do a lot of recomputations with one of the processors, instead of using I/O steps.

\section{Inapproximability}
\label{sec:inapx}

Given the NP-hardness of MPP, a natural follow-up question is whether the optimum can be approximated to some factor in polynomial time. We show that MPP is also hard from this perspective.

\begin{theorem} \label{th:apx}
MPP is APX-hard: there is a constant $\delta>0$ such that no polynomial-time algorithm can approximate the optimum to a $(1+\delta)$ factor, unless $P=NP$.
\end{theorem}

We prove this property already for the simplest case of $k=1$, i.e., SPP with computation costs. The proof can then easily be extended to any $k$ value.

\begin{lemma}
SPP with computation costs is APX-hard.
\end{lemma}

\renewcommand*{\proofname}{Proof sketch}

\begin{proof}
Our proof is motivated by a construction in \cite{RBpebbling3}, which reduces one-shot SPP to vertex cover on a graph of $N$ nodes to show another property. However, this proof considers SPP, where (R3-S) compute steps are free, and hence the proof is rather careless with the number of nodes, using very large gadgets to ensure asymptotic behavior. As a result, the I/O cost is only $\Theta(N)=o(n)$. With computation costs adding up to at least $n$ in SPP, this makes the I/O costs asymptotically irrelevant. Hence the reduction does not work directly for pebbling models with computation cost.

To adapt this idea, we substantially decrease the size of gadgets to ensure that $n=O(N)$, and hence the I/O cost of any solution is in $\Theta(N)=\Theta(n)$. By design, in any reasonable solution, a specific part of this I/O cost is proportional to the size of a vertex cover in the underlying graph; as such, an approximation for the best pebbling could be transformed into an approximation of vertex cover. The main ingredients of the proof are as follows:
\begin{itemize}[topsep=3pt, itemsep=-2pt]
\item we show how to modify the node gadgets in the construction of \cite{RBpebbling3} to reduce their size to constant, thus ensuring that our construction has $n=\Theta(N)$ altogether,
\item we execute some further changes in the gadgets to avoid some other undesired properties, which could be ignored in the original proof of \cite{RBpebbling3} due to their asymptotic analysis,
\item we then show that in any reasonable pebbling, the I/O cost has a linear component that is proportional to the size of a vertex cover in the underlying graph.
\end{itemize}
Altogether, this modified construction allows for an L-reduction to the vertex cover problem in $3$-regular graphs, which is known to be APX-hard \cite{vcapprox}. This implies that it is already NP-hard to find a pebbling of cost at most $(1+\delta) \cdot \texttt{OPT}$ for an appropriate $\delta>0$. 
\end{proof}

In general, the approximability of MPP is an intricate question, as the total cost is often dominated by computation costs, which are at least $\frac{n}{k}$. Thus, even if we have two solutions where I/O costs differ by a large factor, this can translate to only a negligible difference in total cost. To better capture these differences, we introduce an alternative cost metric.

\begin{definition}
Given an MPP pebbling strategy of cost $C$, its \emph{surplus cost} is $C-\frac{n}{k}$.
\end{definition}

Intuitively speaking, surplus cost ignores this unavoidable cost of $\frac{n}{k}$, and instead only measures ``imperfections'' in the pebbling, such as I/O steps, work imbalance between the processors, and recomputations. In terms of this new metric, finding a good solution is much more challenging: it is not possible to approximate the surplus cost in MPP to any finite factor. We again show this via proving an inapproximability result in standard one-shot SPP (i.e.\ minimizing I/O costs, with computations being free): we show that it is already NP-hard here to distinguish the cases when $\texttt{OPT}=0$ and when $\texttt{OPT} \geq n^{1-\varepsilon}$, for any $\varepsilon>0$.

\begin{theorem} \label{th:inapprox}
In one-shot SPP, it is NP-hard to approximate the optimum to any finite multiplicative factor, and to any additive $n^{1-\varepsilon}$ term, in polynomial time (for any $\varepsilon>0$).
\end{theorem}

\begin{proof}
The main part of our proof is to develop a \cdag construction where it is already NP-hard to decide whether one-shot SPP has $\texttt{OPT}=0$ or $\texttt{OPT} \geq 1$. After this, some technical steps allow for extending the same result to an optimum of either $0$ or $n^{1-\varepsilon}$.

When looking for a solution of cost $0$, pebbling becomes notably simpler: blue pebbles cannot be used at all, since (R1-S) and (R2-S) incur cost. The use of (R4-S) is also simple: without recomputation, a red pebble should be deleted exactly when all out-neighbors have been pebbled. As such, a pebbling is characterized by the order of the $n$ computation steps.

We point out that pebbling with only (R3-S) and (R4-S) is essentially equivalent to one-shot \emph{black pebbling}, which is long known to be NP-hard \cite{blackNPhard}; as such, the novelty of this part of our theorem is somewhat limited. However, \cite{blackNPhard} considers a slightly different variant of pebbling, where compute steps can also ``slide'' a pebble from an in-neighbor. It may also be possible to adapt the proof in \cite{blackNPhard} to our case with further work; instead, we present a novel, somewhat simpler reduction based on a different problem, and we also devise new gadgets that might be of independent interest for future works on pebbling.

Our construction is organized into consecutive chains of \emph{level gadgets} that form \emph{towers}. Intuitively, a pebbling strategy always needs to keep one level of each tower in fast memory, and it can `proceed’ to the next level, computing the nodes of the next level and deleting pebbles from the current level. There is no benefit to having partially pebbled levels in our DAG, and hence intuitively, each level can be considered a single entity in our analysis. To pebble the DAG correctly, we need to go through the towers of levels in a carefully designed order, to ensure that the current set of levels never requires more than $r$ pebbles altogether. The concrete level gadgets are shown in Figure \ref{fig:level}, and discussed in detail in Appendix \ref{app:inapprox}.

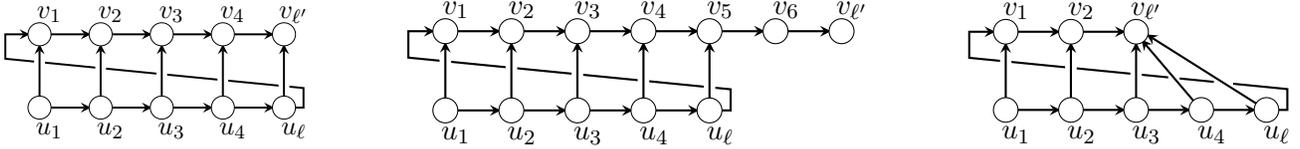
\begin{figure}[t!]
    \vspace{-7pt}
    \centering
    \begin{minipage}{.25\textwidth}
      \centering
      \resizebox{1.0\textwidth}{!}{
      \begin{tikzpicture}

    \draw[thick, arrows=-stealth] (100pt,0pt) -- (108pt,0pt) -- (108pt,8pt) -- (-14pt,20pt) -- (-14pt,30pt) -- (-4pt,30pt);

    \draw[white, fill=white] (-2pt,8pt) rectangle (2pt,22pt);
    \draw[white, fill=white] (23pt,8pt) rectangle (27pt,22pt);
    \draw[white, fill=white] (48pt,8pt) rectangle (52pt,22pt);
    \draw[white, fill=white] (73pt,8pt) rectangle (77pt,22pt);
    \draw[white, fill=white] (98pt,8pt) rectangle (102pt,22pt);

    \begin{scope}[thick, arrows=-stealth]
    \draw (0pt,0pt) -- (21pt,0pt);
    \draw (25pt,0pt) -- (46pt,0pt);
    \draw (50pt,0pt) -- (71pt,0pt);
    \draw (75pt,0pt) -- (96pt,0pt);
    
    \draw (0pt,30pt) -- (21pt,30pt);
    \draw (25pt,30pt) -- (46pt,30pt);
    \draw (50pt,30pt) -- (71pt,30pt);
    \draw (75pt,30pt) -- (96pt,30pt);

    \draw (0pt,0pt) -- (0pt,26pt);
    \draw (25pt,0pt) -- (25pt,26pt);
    \draw (50pt,0pt) -- (50pt,26pt);
    \draw (75pt,0pt) -- (75pt,26pt);
    \draw (100pt,0pt) -- (100pt,26pt);

    \end{scope}

    \draw[black, fill=white] (0pt,0pt) circle (1ex);
    \draw[black, fill=white] (25pt,0pt) circle (1ex);
    \draw[black, fill=white] (50pt,0pt) circle (1ex);
    \draw[black, fill=white] (75pt,0pt) circle (1ex);
    \draw[black, fill=white] (100pt,0pt) circle (1ex);

    \draw[black, fill=white] (0pt,30pt) circle (1ex);
    \draw[black, fill=white] (25pt,30pt) circle (1ex);
    \draw[black, fill=white] (50pt,30pt) circle (1ex);
    \draw[black, fill=white] (75pt,30pt) circle (1ex);
    \draw[black, fill=white] (100pt,30pt) circle (1ex);

    \node[anchor=center] at (4pt,-9.5pt) {$u_1$};
    \node[anchor=center] at (29pt,-9.5pt) {$u_2$}; 
    \node[anchor=center] at (54pt,-9.5pt) {$u_3$}; 
    \node[anchor=center] at (79pt,-9.5pt) {$u_4$}; 
    \node[anchor=center] at (104pt,-9.5pt) {$u_{\ell}$}; 

    \node[anchor=center] at (4pt,38pt) {$v_1$};
    \node[anchor=center] at (29pt,38pt) {$v_2$}; 
    \node[anchor=center] at (54pt,38pt) {$v_3$}; 
    \node[anchor=center] at (79pt,38pt) {$v_4$}; 
    \node[anchor=center] at (104pt,38pt) {$v_{\ell'}$};

\end{tikzpicture}}
    \end{minipage}
    \hspace{.04\textwidth}
    \begin{minipage}{.37\textwidth}
      \centering
      \resizebox{1.0\textwidth}{!}{
      \begin{tikzpicture}

    \draw[thick, arrows=-stealth] (100pt,0pt) -- (108pt,0pt) -- (108pt,8pt) -- (-14pt,20pt) -- (-14pt,30pt) -- (-4pt,30pt);

    \draw[white, fill=white] (-2pt,8pt) rectangle (2pt,22pt);
    \draw[white, fill=white] (23pt,8pt) rectangle (27pt,22pt);
    \draw[white, fill=white] (48pt,8pt) rectangle (52pt,22pt);
    \draw[white, fill=white] (73pt,8pt) rectangle (77pt,22pt);
    \draw[white, fill=white] (98pt,8pt) rectangle (102pt,22pt);

    \begin{scope}[thick, arrows=-stealth]
    \draw (0pt,0pt) -- (21pt,0pt);
    \draw (25pt,0pt) -- (46pt,0pt);
    \draw (50pt,0pt) -- (71pt,0pt);
    \draw (75pt,0pt) -- (96pt,0pt);
    
    \draw (0pt,30pt) -- (21pt,30pt);
    \draw (25pt,30pt) -- (46pt,30pt);
    \draw (50pt,30pt) -- (71pt,30pt);
    \draw (75pt,30pt) -- (96pt,30pt);
    \draw (100pt,30pt) -- (121pt,30pt);
    \draw (125pt,30pt) -- (146pt,30pt);

    \draw (0pt,0pt) -- (0pt,26pt);
    \draw (25pt,0pt) -- (25pt,26pt);
    \draw (50pt,0pt) -- (50pt,26pt);
    \draw (75pt,0pt) -- (75pt,26pt);
    \draw (100pt,0pt) -- (100pt,26pt);

    \end{scope}

    \draw[black, fill=white] (0pt,0pt) circle (1ex);
    \draw[black, fill=white] (25pt,0pt) circle (1ex);
    \draw[black, fill=white] (50pt,0pt) circle (1ex);
    \draw[black, fill=white] (75pt,0pt) circle (1ex);
    \draw[black, fill=white] (100pt,0pt) circle (1ex);

    \draw[black, fill=white] (0pt,30pt) circle (1ex);
    \draw[black, fill=white] (25pt,30pt) circle (1ex);
    \draw[black, fill=white] (50pt,30pt) circle (1ex);
    \draw[black, fill=white] (75pt,30pt) circle (1ex);
    \draw[black, fill=white] (100pt,30pt) circle (1ex);
    \draw[black, fill=white] (125pt,30pt) circle (1ex);
    \draw[black, fill=white] (150pt,30pt) circle (1ex);

    \node[anchor=center] at (4pt,-9.5pt) {$u_1$};
    \node[anchor=center] at (29pt,-9.5pt) {$u_2$}; 
    \node[anchor=center] at (54pt,-9.5pt) {$u_3$}; 
    \node[anchor=center] at (79pt,-9.5pt) {$u_4$}; 
    \node[anchor=center] at (104pt,-9.5pt) {$u_{\ell}$}; 

    \node[anchor=center] at (4pt,38pt) {$v_1$};
    \node[anchor=center] at (29pt,38pt) {$v_2$}; 
    \node[anchor=center] at (54pt,38pt) {$v_3$}; 
    \node[anchor=center] at (79pt,38pt) {$v_4$}; 
    \node[anchor=center] at (104pt,38pt) {$v_5$};
    \node[anchor=center] at (129pt,38pt) {$v_6$}; 
    \node[anchor=center] at (154pt,38pt) {$v_{\ell'}$};

\end{tikzpicture}}
    \end{minipage}
    \hspace{.04\textwidth}
    \begin{minipage}{.265\textwidth}
      \centering
      \resizebox{1.0\textwidth}{!}{
      \begin{tikzpicture}

    \draw[thick, arrows=-stealth] (100pt,0pt) -- (108pt,0pt) -- (108pt,8pt) -- (-14pt,20pt) -- (-14pt,30pt) -- (-4pt,30pt);

    \draw[white, fill=white] (-2pt,8pt) rectangle (2pt,22pt);
    \draw[white, fill=white] (23pt,8pt) rectangle (27pt,22pt);
    \draw[white, fill=white] (48pt,8pt) rectangle (52pt,22pt);
    \draw[white, fill=white] (62.5pt,8pt) rectangle (66.5pt,22pt);
    \draw[white, fill=white] (80pt,8pt) rectangle (87pt,22pt);

    \begin{scope}[thick, arrows=-stealth]
    \draw (0pt,0pt) -- (21pt,0pt);
    \draw (25pt,0pt) -- (46pt,0pt);
    \draw (50pt,0pt) -- (71pt,0pt);
    \draw (75pt,0pt) -- (96pt,0pt);
    
    \draw (0pt,30pt) -- (21pt,30pt);
    \draw (25pt,30pt) -- (46pt,30pt);

    \draw (0pt,0pt) -- (0pt,26pt);
    \draw (25pt,0pt) -- (25pt,26pt);
    \draw (50pt,0pt) -- (50pt,26pt);
    \draw (75pt,0pt) -- (52pt,27pt);
    \draw (100pt,0pt) -- (54pt,29pt);

    \end{scope}

    \draw[black, fill=white] (0pt,0pt) circle (1ex);
    \draw[black, fill=white] (25pt,0pt) circle (1ex);
    \draw[black, fill=white] (50pt,0pt) circle (1ex);
    \draw[black, fill=white] (75pt,0pt) circle (1ex);
    \draw[black, fill=white] (100pt,0pt) circle (1ex);

    \draw[black, fill=white] (0pt,30pt) circle (1ex);
    \draw[black, fill=white] (25pt,30pt) circle (1ex);
    \draw[black, fill=white] (50pt,30pt) circle (1ex);

    \node[anchor=center] at (4pt,-9.5pt) {$u_1$};
    \node[anchor=center] at (29pt,-9.5pt) {$u_2$}; 
    \node[anchor=center] at (54pt,-9.5pt) {$u_3$}; 
    \node[anchor=center] at (79pt,-9.5pt) {$u_4$}; 
    \node[anchor=center] at (104pt,-9.5pt) {$u_{\ell}$}; 

    \node[anchor=center] at (4pt,38pt) {$v_1$};
    \node[anchor=center] at (29pt,38pt) {$v_2$}; 
    \node[anchor=center] at (54pt,38pt) {$v_{\ell'}$};

\end{tikzpicture}}
    \end{minipage}
    \vspace{-7pt}
    \caption{Examples of consecutive levels. The 1st level (bottom row) always has size $\ell=5$. The 2nd level (top row) has $\ell'=5$ on the left side, $\ell'=7$ in the middle, and $\ell'=3$ on the right.}
    \label{fig:level}
\end{figure}

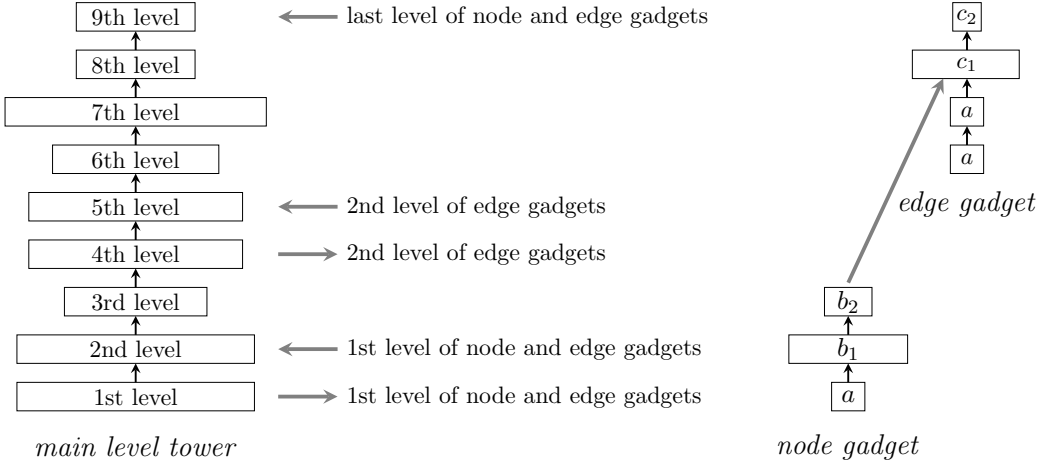
\begin{figure}
\vspace{7pt}
    \centering
    \resizebox{0.8\textwidth}{!}{
    \begin{tikzpicture}

    \node[anchor=center] at (0pt,-15pt) {\large \textit{main level tower}};

    \draw[black, fill=white] (-50pt,0pt) rectangle (50pt,12pt);
    \draw[black, fill=white] (-50pt,20pt) rectangle (50pt,32pt);
    \draw[black, fill=white] (-30pt,40pt) rectangle (30pt,52pt);
    \draw[black, fill=white] (-45pt,60pt) rectangle (45pt,72pt);
    \draw[black, fill=white] (-45pt,80pt) rectangle (45pt,92pt);
    \draw[black, fill=white] (-35pt,100pt) rectangle (35pt,112pt);
    \draw[black, fill=white] (-55pt,120pt) rectangle (55pt,132pt);
    \draw[black, fill=white] (-25pt,140pt) rectangle (25pt,152pt);
    \draw[black, fill=white] (-25pt,160pt) rectangle (25pt,172pt);

    \node[anchor=center] at (0pt,6pt) {\small 1st level};
    \node[anchor=center] at (0pt,26pt) {\small 2nd level};
    \node[anchor=center] at (0pt,46pt) {\small 3rd level};
    \node[anchor=center] at (0pt,66pt) {\small 4th level};
    \node[anchor=center] at (0pt,86pt) {\small 5th level};
    \node[anchor=center] at (0pt,106pt) {\small 6th level};
    \node[anchor=center] at (0pt,126pt) {\small 7th level};
    \node[anchor=center] at (0pt,146pt) {\small 8th level};
    \node[anchor=center] at (0pt,166pt) {\small 9th level};

    \begin{scope}[thick, arrows=-stealth]
    \draw (0pt,12pt) -- (0pt,20pt);
    \draw (0pt,32pt) -- (0pt,40pt);
    \draw (0pt,52pt) -- (0pt,60pt);
    \draw (0pt,72pt) -- (0pt,80pt);
    \draw (0pt,92pt) -- (0pt,100pt);
    \draw (0pt,112pt) -- (0pt,120pt);
    \draw (0pt,132pt) -- (0pt,140pt);
    \draw (0pt,152pt) -- (0pt,160pt);
    \end{scope}

    \draw[ultra thick, gray, arrows=-stealth] (60pt,6pt) -- (85pt,6pt);
    \node[anchor=west] at (85pt,6pt) {\small 1st level of node and edge gadgets};
    \draw[ultra thick, gray, arrows=-stealth] (85pt,26pt) -- (60pt,26pt);
    \node[anchor=west] at (85pt,26pt) {\small 1st level of node and edge gadgets};
    \draw[ultra thick, gray, arrows=-stealth] (60pt,66pt) -- (85pt,66pt);
    \node[anchor=west] at (85pt,66pt) {\small 2nd level of edge gadgets};
    \draw[ultra thick, gray, arrows=-stealth] (85pt,86pt) -- (60pt,86pt);
    \node[anchor=west] at (85pt,86pt) {\small 2nd level of edge gadgets};
    \draw[ultra thick, gray, arrows=-stealth] (85pt,166pt) -- (60pt,166pt);
    \node[anchor=west] at (85pt,166pt) {\small last level of node and edge gadgets};

    \node[anchor=center] at (300pt,-15pt) {\large \textit{node gadget}};

    \draw[black, fill=white] (293pt,00pt) rectangle (307pt,12pt);
    \draw[black, fill=white] (275pt,20pt) rectangle (325pt,32pt);
    \draw[black, fill=white] (290pt,40pt) rectangle (310pt,52pt);

    \node[anchor=center] at (300pt,6pt) {$a$};
    \node[anchor=center] at (300pt,26pt) {$b_1$};
    \node[anchor=center] at (300pt,46pt) {$b_2$};
    
    \begin{scope}[thick, arrows=-stealth]
    \draw (300pt,12pt) -- (300pt,20pt);
    \draw (300pt,32pt) -- (300pt,40pt);
    \end{scope}

    \node[anchor=center] at (350pt,87pt) {\large \textit{edge gadget}};

    \draw[black, fill=white] (343pt,100pt) rectangle (357pt,112pt);
    \draw[black, fill=white] (343pt,120pt) rectangle (357pt,132pt);
    \draw[black, fill=white] (327pt,140pt) rectangle (372pt,152pt);
    \draw[black, fill=white] (344pt,160pt) rectangle (356pt,172pt);

    \node[anchor=center] at (350pt,106pt) {\small $a$};
    \node[anchor=center] at (350pt,126pt) {\small $a$};
    \node[anchor=center] at (350pt,146pt) {\small $c_1$};
    \node[anchor=center] at (350pt,166pt) {\small $c_2$};
    
    \begin{scope}[thick, arrows=-stealth]
    \draw (350pt,112pt) -- (350pt,120pt);
    \draw (350pt,132pt) -- (350pt,140pt);
    \draw (350pt,152pt) -- (350pt,160pt);
    \end{scope}

    \draw[ultra thick, gray, arrows=-stealth] (300pt,54pt) -- (340pt,140pt);

\end{tikzpicture}}
    \vspace{-5pt}
    \caption{High-level sketch of our construction: the main tower on the left, and node/edge gadgets on the right, with the level sizes and the dependency between incident node-edge pairs also shown.}
    \label{fig:inapprox}
    \vspace{-5pt}
\end{figure}

Using these gadgets, we present a reduction from finding a clique of size $q$ in a graph $G'$. Our DAG has a main tower, which is used to control the number of available pebbles at all times. We also add smaller tower gadgets for each node and edge of $G'$, as illustrated in Figure \ref{fig:inapprox}. Furthermore, we draw edges between specific levels of the node and edge gadgets whenever a node is incident to an edge, and between specific levels of the main tower and all node and edge gadgets.

In the beginning, we can only pebble the 1st level of the main tower, then the first level of each node/edge gadget, and then the 2nd level of the main tower. Proceeding to the 3rd level of the main tower then frees up many pebbles; we can use these to proceed to the 2nd (and 3rd) level of node gadgets. However, due to $b_2>a$, we can only move to the 3rd level of at most $q$ node gadgets; otherwise, we would not have enough remaining pebbles to reach the 4th level of the main tower afterwards. Thus intuitively, any pebbling strategy corresponds to selecting at most $q$ node gadgets.

The next levels of the main tower allow/require us to move to the 2nd level in edge gadgets. The 6th level again provides more free pebbles, so we can move to the 3rd (and then 4th) level in any edge gadget, provided that both of its incident nodes were chosen earlier. The next (7th) level then allows the fewest free pebbles: we can only proceed to this level if we reached the 4th level in at least ${q \choose 2}$ edge gadgets. This is only possible if we found a set of $q$ nodes in $G'$ that span ${q \choose 2}$ edges. Finally, after this point, the pebbling is easy to finish. Hence a pebbling of cost $0$ exists if and only if $G'$ has a clique of size $q$.
\end{proof}

The theorem directly carries over to surplus cost in MPP: given the same DAG in MPP, a pebbling of cost $0$ will translate to a surplus cost of $0$, while a pebbling with $n^{1-\varepsilon}$ I/O steps (or the same amount of recomputation) will translate to a surplus cost of at least $n^{1-\varepsilon}$.

\begin{corollary} \label{cor:surplus_cost}
In MPP, it is NP-hard to approximate the optimum surplus cost to any finite multiplicative factor, and to any additive $n^{1-\varepsilon}$ term, in polynomial time (for any $\varepsilon>0$).
\end{corollary}

\newpage

\bibliography{references}

\newpage

\appendix

\section{Further model discussion} \label{app:models}

\subsection{Other Model Variations} \label{sec:model_variants}

We briefly discuss several (minor) variations of MPP. Most of our results are easily adapted to these slightly different settings.

Similarly to SPP, there are multiple ways to define the initial and terminal configurations in MPP. For instance, instead of having source nodes that are freely computable with \ref{rule:multi-compute}, another SPP variant assumes that source nodes cannot be computed at all; instead, they have a blue pebble initially, and their data needs to be loaded with \ref{rule:multi-bluebyred}. Similarly, instead of requiring all sink nodes to have a pebble of any color in the end, some SPP variants explicitly require that 
all sink nodes must have a blue pebble (i.e., sinks must all be saved to slow memory). Since the majority of our constructions have only $O(1)$ sources and sinks, these changes only add a $O(1)$ term to the cost of any pebbling.

Another alternative for MPP is to assume that processors cannot only exchange data via the slow memory, but also directly with each other, i.e., replacing a red pebble of $p_1$ directly by a red pebble of $p_2$. For such an \textit{MPP with direct sending}, we can replace the rules \ref{rule:multi-redbyblue} and \ref{rule:multi-bluebyred} by a new rule (R1-M)* which allows us to take any set of different-colored pebbles (blue and/or different shades of red), and replace it by any other set of different-colored pebbles; we define this formally in Appendix \ref{app:models}. An application of (R1-M)* corresponds to a single communication step in a fully connected network topology between the $k$ processors and the slow memory, with every participant sending and receiving at most one value. In the simplest case of $r=\infty$, this is essentially equivalent to DAG scheduling in the BSP model without latency \cite{BSPscheduling}. This demonstrates that with minor variations, MPP is indeed a generalization of both SPP and DAG scheduling problems. 

Formally, in the MPP variant with \emph{direct sending}, for $m \in \mathbb{Z}$ with $m \leq k+1$, let us define an \emph{extended \colorful selection} as a injective function $f_m: [m] \to [k] \cup \{ 0 \}$. Instead of rules \ref{rule:multi-redbyblue} and \ref{rule:multi-bluebyred}, we only have a single I/O rule, as follows:
\begin{enumerate}[label=(R\arabic*-M*), align=right, leftmargin=4.5em]
 \item Consider an extended \colorful selections $f_m$, and vertices $v_1,v_2, \ldots, v_m$ such that $f_m$ satisfies $v_i \! \in \! R^{f_{m}(i)}$ if $f_{m}(i) \neq 0$ and $v_i \! \in \! B$ if $f_{m}(i) = 0$, for all $i \! \in \! [m]$. For some other extended \colorful selection $f'_m$ with the same $m$, add $v_i$ to $R^{f'_{m}(i)}$ if $f'_{m}(i) \neq 0$, and add $v_i$ to $B$ if $f'_{m}(i) = 0$, for all $i \in [m]$.
\end{enumerate}

The remaining components of the MPP definition remain unchanged, including rules \ref{rule:multi-compute} and \ref{rule:multi-remove}. As before, rule \ref{rule:multi-compute} incurs a cost of $1$, while the new rule (R1-M)* incurs a cost of $g$.

This modified I/O rule essentially assumes a fully connected (clique-like) communication network between the $k$ processors and the slow memory, where in a single I/O step, any of the endpoints in this network (i.e., processor or slow memory) can send a data value to any other endpoint, but all endpoints can only send at most one and receive at most one piece of data. Several parallel computing models assume such an underlying communication setup, such as the communication cost function (named $h$-relation) in the BSP model \cite{BSPbook1, BSPbook2}.

This variant is relevant because with direct sending, MPP essentially becomes a generalization of the DAG scheduling problem in BSP without latency, analyzed in \cite{BSPscheduling}. This scheduling problem considers a more realistic model than previous works on DAG scheduling (already assuming e.g., communication costs that scale with data volume), but still makes some simplifications compared to practice (e.g., ignoring latency, assuming uniform node weights) for a simpler theoretical analysis. A BSP schedule in this model consists of computation and communication phases; however, these can all be split into separate steps for MPP with direct sending such that the total cost is unchanged. In particular, the cost $s_{comp}$ of a computation phase is the maximum number of nodes executed on any processor; we can indeed split any such phase into $s_{comp}$ distinct \ref{rule:multi-compute} steps. The cost $g \cdot s_{comm}$ of a communication phase is defined via $h$-relations; it is already discussed in \cite{BSPscheduling} that any such phase can be split into $s_{comm}$ distinct communication rounds, each of which are equivalent to a use of (R1-M)*. For more details on this scheduling model, we refer the reader to \cite{BSPscheduling}.

Thus in the special case of $r\!=\!\infty$, MPP with direct sending is almost equivalent to a DAG scheduling problem. We note that the equivalence is not complete: even with $r\!=\!\infty$, an MPP strategy might still prefer to use blue pebbles. For a simple example, consider a pebbling strategy with two (R1-M)* steps and three processors, where (i) $p_1$ and $p_2$ need to exchange a pair of node data in the first (R1-M)* step, (ii) $p_2$ and $p_3$ need to exchange a pair of node data in the second (R1-M)* step, and (iii) $p_3$ needs to send another node $v$ to $p_1$ throughout this process. Without blue pebbles, we cannot send $v$ in either of the steps, since $p_1$ is already receiving in the first step, and $p_3$ is already sending in the second. However, with blue pebbles also, we can move $v$ from $p_3$ to slow memory in the first (R1-M)* step, and from slow memory to $p_1$ in the second (R1-M)* step. As such, the optimal cost in direct sending MPP might be lower than the optimum of the corresponding DAG scheduling problem.

It is also an interesting question how this direct sending variant differs from our original MPP problem. We discuss this in more detail in Appendix \ref{app:NPhard} when analyzing how earlier NP-hardness results from DAG scheduling carry over to MPP.

We note that it would also be interesting to consider an MPP variant where a (R1-M)* step is \emph{added} to the set of rules, instead of replacing \ref{rule:multi-redbyblue} and \ref{rule:multi-bluebyred}. Such a model would allow for a natural separation of I/O steps, based on whether they are motivated by memory limitations or by inter-processor communication. We leave it to future work to analyze such a model extension.

\subsection{Comparison to constrained model} \label{sec:constrained_MPP}

The closest previous work to our model is the extension of SPP to multiple processors in \cite{RBpebbling5}. In this model, each processor can simultaneously execute a single one of the SPP rules, possibly different ones. The goal of the problem is to minimize the total number of applications of rules (R1-S) and (R2-S) in a pebbling; however, the strategy space is restricted to pebblings where the length of the sequence is minimal. This model differs from MPP in several ways:
\begin{itemize}[leftmargin=1.2em]
\item The model in \cite{RBpebbling5} assumes each processor concurrently executes one of (R1-S)--(R4-S) in each step, hence implicitly assuming that all SPP steps take the same time.
\item The model aims to minimize the total number of occurrences of (R1-S) and (R2-S) in a pebbling. This implicitly assumes that I/O on different processors cannot be parallelized.
\item To motivate the parallelization of compute steps, the model has a further constraint: it only allows pebbling strategies that consist of the minimal possible number of steps. However, it is not even clear whether this condition can be verified in polynomial time.
\end{itemize}

Moreover, the strict constraint on the length of the pebbling can often exclude solutions with much fewer I/O steps; as such, while this is a promising preliminary model, we believe that our work is a simpler and more realistic way to capture red-blue pebbling in a parallel setting.

For completeness, we show a simple example where this artificially constrained model only allows pebblings with much higher I/O cost. Consider a DAG on $6$ nodes: two chains of length $2$ and $3$, respectively, with their first node having an incoming edge from a common source node. Assume $k\!=\!2$ and $r\! \geq \! 6$. If executed on a single processor, this sequence consists of $6$ compute steps, with no I/O operations. On the other hand, the pebbling sequence of minimal length also uses I/O steps: $p_1$ computes the source node, then sends it to slow memory via (R1-S). For the next 3 steps, $p_1$ computes the chain of length $3$, while $p_2$ loads the source node from slow memory, and computes the chain of length $2$. This pebbling sequence consists of only $5$ steps, but has $2$ I/O steps. If $g$ is large (i.e., in any model that accounts for the fact that communication steps are much more costly in practice than compute steps), this is a significantly more costly strategy than the one with $6$ compute steps. However, since the constrained model only allows sequences of minimal length, it is restricted to this more expensive solution.

We can easily construct similar examples on arbitrary large DAGs: consider two sets $A$ and $B$ of $m$ nodes each, and two sets $C_1$ and $C_2$ of $(2m\!+\!1)$ and $(m\!+\!1)$ nodes, respectively, with $(u,v) \in E$ for all $u \in A$, $v \in C_1 \cup C_2$. If $A$, $C_1$ and $C_2$ are all computed on $p_1$ (and $B$ on $p_2$), we get a pebbling sequence of $4m+2$ compute steps, without any I/O. However, if $A$ and $C_1$ are computed on $p_1$, while $B$ and $C_2$ on $p_2$, then we need only need $4m+1$ steps in our sequence, but we have $2m$ I/O operations (of sending the data of $A$ to $p_2$).

\section{Proofs for properties in Section~\ref{sec:basics}} \label{app:basics}

\subsection{NP-hardness: proof of Lemma \ref{lem:NPhard}} \label{app:NPhard}

We now discuss how to apply the constructions in \cite{BSPscheduling} to show that MPP is already NP-hard on 2-layer DAGs and in-trees. Note that the work of \cite{BSPscheduling} considers DAG scheduling in the BSP model. As discussed in Appendix \ref{app:models}, this is equivalent to the direct sending variant of MPP with $r=\infty$ and the modification that we cannot use blue pebbles at all. As such, for the rest of this section, we will understand the BSP scheduling problem to refer to this direct sending MPP problem variant with $r=\infty$ and no use of blue pebbles.

\subsubsection{Original MPP vs. Direct Sending}

Consider the BSP scheduling problem with some parameter $g$, and our original MPP problem with $g'=\frac{g}{2}$ and $r=\infty$. At first glance, it might seem that these two settings are equivalent, since any direct sending (i.e., (R1-M)* communication step) in BSP scheduling can be replaced by a consecutive \ref{rule:multi-redbyblue} and \ref{rule:multi-bluebyred} step, by sending the same (at most $k$) data values first all to slow memory, then to the corresponding processors. Indeed, the existence of a BSP schedule of cost $C$ implies that there also exists an MPP pebbling with the same cost $C$. However, the reverse is not true: the optimum cost in MPP pebbling can be lower than in the corresponding scheduling problem.

Intuitively, there are two ways in which MPP with $g'=\frac{g}{2}$ and $r=\infty$ is a more flexible model than BSP scheduling with $g$. On the one hand, splitting each communication step into an \ref{rule:multi-redbyblue} and \ref{rule:multi-bluebyred} part allows for efficient broadcast operations: if a node $v$ has to be sent from processor $p_1$ to all the remaining $(k-1)$ processors, then this can be achieved in a single pair of \ref{rule:multi-redbyblue} and \ref{rule:multi-bluebyred} steps in MPP (at a cost of $g'+g'=g$), whereas it takes $(k-1)$ communication steps in BSP scheduling (and a cost of $(k-1) \cdot g$).

On the other hand, the splitting of communication steps into an \ref{rule:multi-redbyblue} and an \ref{rule:multi-bluebyred} part also allows us to better parallelize communications occasionally. For instance, assume that we want to send a node $v_1$ from processor $p_1$ to $p_3$ at some point in time $t_1$, and another node $v_2$ from $p_2$ to $p_1$ at a later time $t_2$, and finally, there is a third data $v_3$ to be sent from $p_2$ to $p_3$, which can be sent either at $t_1$ or $t_2$ (i.e., $v_3$ is computed on $p_2$ before $t_1$, but only required on $p_3$ after $t_2$). In BSP scheduling, the transfer of $v_3$ cannot be parallelized with either of the other operations, since $p_3$ is already receiving a value at $t_1$, and $p_2$ is already sending a value at $t_2$; as such, it requires a separate communication step, and hence the total cost of the communication steps is $3 \cdot g$. On the other hand, in MPP, we can already move $v_3$ to slow memory in the \ref{rule:multi-redbyblue} step at $t_1$ (together with $v_1$), and load it from slow memory in the \ref{rule:multi-bluebyred} step at $t_2$ (together with $v_2$), since $p_2$ is not yet sending a value at $t_1$, and $p_3$ is not yet receiving a value at $t_2$. The total cost is then only $4 \cdot g' = 2 \cdot g$.

As such, in order to show that the NP-hardness proofs for BSP scheduling from \cite{BSPscheduling} also carry over to MPP, we need to show that the differences above do not affect the optimum pebbling cost in the specific constructions, and hence an MPP scheduling of a given cost $C$ exists if and only if a BSP scheduling of the same cost exists.

\subsubsection{2-layer DAGs}

For $2$-layer DAGs, we use the fact that the proof reduction in \cite{BSPscheduling} only has $k=2$; we show that for this case, the optimum costs in the corresponding BSP scheduling and MPP problems are always identical. We show that any MPP pebbling can be turned into a BSP schedule of at most the same cost (the reverse was already shown above). Note that for $k=2$ and $r=\infty$, all data values sent from $p_1$ are communicated in order to be received by $p_2$, and all data values sent from $p_2$ are communicated in order to be received by $p_1$. We assume that our starting MPP schedule is reasonable, i.e., it does not send values to slow memory that are not loaded later by the other processor, or does not send/load the same value from/to a processor multiple times; otherwise, we can remove the corresponding unnecessary steps from the pebbling without increasing the cost.

We will consider the \ref{rule:multi-redbyblue} and \ref{rule:multi-bluebyred} operations in the MPP solution in order, and replace them by (R1-M)* steps without ever increasing the total cost. In particular, consider the first \ref{rule:multi-bluebyred} step in the pebbling that has not been converted yet. First let us assume that this step only loads a single value from slow memory (assume, w.l.o.g., that this is a node $v_1$ sent from $p_1$ to $p_2$). Let us find the preceding \ref{rule:multi-redbyblue} step that saves $v_1$ from $p_1$ to slow memory. If this \ref{rule:multi-redbyblue} step only transfers $v_1$, then we can combine the corresponding \ref{rule:multi-redbyblue} and \ref{rule:multi-bluebyred} into a (R1-M)* step, sending $v_1$ from $p_1$ to $p_2$, at the original time of the \ref{rule:multi-bluebyred} step, with the cost remaining unchanged ($g'+g'=g$). If, on the other hand, the \ref{rule:multi-redbyblue} step saving $v_1$ also transfers another value $v_2$ from $p_2$ to slow memory, then we can merge the two operations into a BSP communication step that, sends $v_1$ from $p_1$ to $p_2$ and $v_2$ from $p_2$ to $p_1$, at the original time of the \ref{rule:multi-bluebyred} step, again at the same cost. Note that in our original pebbling, $v_2$ was loaded to $p_1$ from slow memory at a later point in time; we can now remove $v_2$ from this \ref{rule:multi-bluebyred} operation, possibly deleting the entire operation if no other data value was loaded to $p_2$ in the same step.

On the other hand, assume that the first unprocessed \ref{rule:multi-bluebyred} step loads both $v_1$ (originally sent from $p_1$) to $p_2$ and $v_2$ (originally sent from $p_2$) to $p_1$. In this case, there are two corresponding \ref{rule:multi-redbyblue} steps earlier that send $v_1$ and $v_2$ to slow memory. If these two are in fact the same \ref{rule:multi-redbyblue} step, then we can simply combine the \ref{rule:multi-redbyblue} and \ref{rule:multi-bluebyred} step into a (R1-M)* step of the same cost, at the original time of the \ref{rule:multi-bluebyred} step. Otherwise, if one or both of these \ref{rule:multi-redbyblue} steps also save other value(s) to slow memory besides $v_1$ and $v_2$, we can still create the same (R1-M)* step sending $v_1$ from $p_1$ to $p_2$ and $v_2$ from $p_2$ to $p_1$, and merge the saving of these possible other values to slow memory into a new \ref{rule:multi-redbyblue} step. Note that besides $v_1$, an \ref{rule:multi-redbyblue} step could only send a value from $p_2$, and besides $v_2$, an \ref{rule:multi-redbyblue} step could only send a value from $p_1$, so the remainder of these \ref{rule:multi-redbyblue} steps can indeed be merged into a new \ref{rule:multi-redbyblue} step. Furthermore, the number of \ref{rule:multi-redbyblue} steps indeed decreases by at least $1$ in this process, and the newly merged \ref{rule:multi-redbyblue} step still results in a valid pebbling if placed at the time of the later one of the original two \ref{rule:multi-redbyblue} steps.

As such, the conversion shows that there exists a direct sending pebbling strategy with at most the same communication cost as our standard MPP pebbling, and hence the two problems have equal optimum values. As a result, the reduction of \cite{BSPscheduling} for NP-hardness also carries over to MPP.

\subsubsection{In-trees}

For in-trees, one needs to delve into the concrete proof details in \cite{BSPscheduling} to show that the same reduction holds for MPP. We outline the minor differences required in the proof when applied to MPP; for more details and context, we refer the reader to the detailed proof in \cite{BSPscheduling}.

On a high level, the reduction has an allowed cost of $\frac{n}{k} + m_0 \cdot g$ for some parameter $m_0$, which immediately implies that we can have at most $m_0$ communication steps. Furthermore, there is a sink node $u_0$; assume w.l.o.g.\ that it is assigned to processor $p_0$. This $u_0$ has $\frac{n}{k}-1$ in-neighbors, and also $m_0$ further incoming edges from nodes $u_i$ in other gadgets; one can show that all of these nodes $u_i$ must all be assigned to a different processor than $p_0$, otherwise we exceed the allowed cost. This means that we need to send $m_0$ distinct data values from other processors to $p_0$; in BSP scheduling, this requires $m_0$ separate communication steps. Since the allowed cost is $\frac{n}{k} + m_0 \cdot g$, this already implies that the cost of computation steps must sum up to $\frac{n}{k}$, and the cost of communication steps must sum up to $ m_0 \cdot g$. The rest of the analysis is based on these observations.

In contrast to this, in MPP with $g'=\frac{g}{2}$, the same setting does not immediately imply that these communication steps require all our budget for communication. In particular, loading the $m_0$ values from slow memory to $p_0$ still must happen in $m_0$ separate \ref{rule:multi-bluebyred} steps, and hence at a cost of $m_0 \cdot g'$. However, the $m_0$ values may, in fact, be saved to slow memory in a parallel fashion from the other $(k-1)$ processors, if the $u_i$ are distributed in a balanced way among them; instead of $m_0 \cdot g'$, this would only incur a cost of $\frac{m_0}{k-1} \cdot g'$. As such, this setting would allow us to have extra computation or communication steps compared to BSP scheduling, without violating the cost limit.

Hence in order to show that the communication steps must still sum up to a cost of $2 \cdot m_0 \cdot g'$, we need to consider further proof details. In particular, to enforce a timing on the communication steps, the construction has a critical path of length $\frac{n}{k}$; if we only have $\frac{n}{k}$ computation steps to fit into the allowed cost, then this implies that the $j$-th node of this path must always be computed in the $j$-th computation step. In MPP, we now only have a slack of $s:=(m_0-\frac{m_0}{k-1}) \cdot g'$, i.e., we have so far only shown an upper bound of $\frac{n}{k} + s$ on the number of allowed computation steps, so with a similar approach, we can only claim that the $j$-th node of the path must be computed in the $(j + s)$-th computation step at the latest.

The key observation for adapting the proof to MPP is that such a small additive term $s$ does not affect the induction step in the proof that establishes the timing of the communication steps, since $s \leq m_0 \cdot g'$, i.e., a much smaller magnitude than the size of gadgets in the construction. In particular, for the first node of the critical path that has a gadget attached (and then for each subsequent such node in the induction), it still holds that it has far too many predecessors to compute on a single processor in the first $(j + s)$ computation steps. Hence as in the original proof, a communication step is required within these gadgets to share the workload, which in MPP means that we need to have both a \ref{rule:multi-redbyblue} and a \ref{rule:multi-bluebyred} step. However, in these first $(j + s)$ compute steps, we still do not have enough remaining computations to finish more than one of the gadgets with $u_i$ that have an edge to the sink $u_0$. As such, any valid pebbling has at least one \ref{rule:multi-redbyblue} and \ref{rule:multi-bluebyred} step before the $(j + s)$-th compute step, and these must transfer exactly one of the data values $u_i$ to slow memory and then to $p_0$. The induction follows for the next splitting points of the critical path in a similar way, as in the original proof: even with the potential $s$ extra compute steps, the remaining compute steps until the given computation deadline are not enough for the remaining predecessors, and hence another communication (pair of \ref{rule:multi-redbyblue} and \ref{rule:multi-bluebyred} steps) is required. This \ref{rule:multi-redbyblue} can once again only upload a single one of the values $u_i$ to slow memory, since the compute steps are not enough to compute more than one of the gadgets connected to $u_0$. Altogether, this induction shows that any valid pebbling must have $m_0$ distinct \ref{rule:multi-redbyblue} steps, and save a single one of the $u_i$ values to slow memory in each of them.

This implies that any valid pebbling will have a total cost of at least $m_0 \cdot 2 \cdot g'= m_0 \cdot g$ for communication, and hence a cost of exactly $\frac{n}{k}$ for computation steps; from here, the proof becomes identical to that in the original setting. For simplicity, we can even delay all the \ref{rule:multi-redbyblue} steps in MPP to immediately before the succeeding \ref{rule:multi-bluebyred}; the pebbling still remains valid, and the resulting sequence now becomes essentially identical to the corresponding BSP schedule, since the \ref{rule:multi-redbyblue}-\ref{rule:multi-bluebyred} pairs can be converted into (R1-M)* steps.

\subsection{Greedy approach: proofs of Lemmas \ref{lem:greedy_upper} and \ref{lem:greedy_lower}}

The proof of the upper bounds on the greedy algorithm were already outlined before. Recall that these follow from a fundamental result in DAG scheduling, which states that in the simplest \cdag scheduling model (which is equivalent to MPP with $g\!=\!0$ and $r\!=\!\infty$), any non-idle greedy strategy (i.e., where no processor remains idle when there are computable nodes) is a $2$-approximation \cite{greedy2approx}.

\renewcommand*{\proofname}{Proof of Lemma \ref{lem:greedy_upper}}

\begin{proof}
Let $L_0$ denote the minimal number of compute steps required in any pebbling; clearly $L_0 \leq \texttt{OPT}$. As noted above, a non-idle greedy strategy has at most $2 \cdot L_0$ compute steps. Using the worst-case I/O strategy discussed for Lemma \ref{lem:trivial_bounds} (always saving every node to, and always loading every in-neighbor from slow memory) adds at most $(\Delta_{in}+1)$ I/O steps to each compute step, and hence its total cost is at most $2 \cdot L_0 \cdot (g \cdot (\Delta_{in}+1)+1)$.
\end{proof}

We now move on to discuss the constructions to show Lemma \ref{lem:greedy_lower}, i.e., the lower bounds on the suboptimality of the greedy pebbling approach.

Recall that our greedy pebbling strategy is not entirely defined; instead, we consider a class of greedy strategies, and our lower bound holds for any strategy from this class. In the first steps, each of the processors selects a distinct source node in the DAG, and computes it concurrently (our constructions both have $k=2$ and exactly two source nodes, so with this all sources are already computed). Then in every subsequent step, we consider the set of \textit{ready nodes}, i.e., nodes $v$ such that all the in-neighbors of $v$ have already been computed before on some processor, but $v$ has not yet been computed. Each processor $p_i$ selects a \textit{next-goal node} from among these ready nodes: according to the greedy rule, this could be either the ready node with the highest number of red pebbles of $p_i$ on its in-neighbors, or alternatively, the highest ratio of red pebbles on in-neighbors to the number of in-neighbors (at the point of computing the previous next-goal node of $p_i$). In case if there are multiple nodes tied for next-goal node, the processors can break ties arbitrarily; if multiple processors would select the same node, we can assign this node to one of these processors (with arbitrary tie-breaking), and the remaining processors can then follow the same rule to choose from the remaining ready nodes. Also, in case of the ratio of red pebbled in-neighbors, the ratio for source nodes can be interpreted as either $0$ or $1$, as desired.

Given the selected next-goal nodes, we assume that all processors compute these nodes concurrently. Processors can then execute an arbitrary amount of steps before this computation step, in order to ensure that in-neighbors of the next-goal node are available on the processor: at any point in the schedule before computing the next-goal node, they can load any node from slow memory, save any node to slow memory, delete any pebbles to save up space, or even recompute any node(s) that were previously computed on the same processor. That is, the strategy only defines the next new node to be computed by each processor, in a greedy manner, but the strategy is free to decide what I/O steps are executed (and when) in order to achieve this, which data values are kept in (or removed from) fast memory, and which data values are obtained by recomputation; as such, it covers a rather general class of greedy methods. Nonetheless, this simple restriction on the greedy selection of the next node to compute is already enough to show our lower bounds, i.e., that any strategy in this class can be rather far from the optimum.

\subsubsection{Lemma \ref{lem:greedy_lower}: lower bound of $\frac{\Delta_{in}-1}{2}$}

For the first claim in Lemma \ref{lem:greedy_lower}, we need to execute two modifications on the zipper gadget of Figure \ref{fig:gadget}. On the one hand, we take two copies of this DAG. On the other hand, we replace the input groups by two gadgets that misguide the greedy heuristic, to ensure that our $k=2$ processors will start working independently on the two copies.

More specifically, let $u_{1}, ..., u_{d}, u_{d+1}, ..., u_{2d}$ denote the nodes in the input groups (with $u_{1}, ..., u_{d}$ in the first, and $u_{d+1}, ..., u_{2d}$ in the second). We still attach a chain of length $2g$ in front of each of these $2d$ nodes as before, in order to ensure that it is not beneficial to recompute them. For all $i \in [2d-1]$, we also add a further outgoing edge from $u_i$ to the first node of the chain attached to $u_{i+1}$.

This ensures that in both copies, there is only a single topological ordering of the input group gadgets, which goes through the $2d$ chains consecutively, following each from front to end. In the greedy algorithm, the two processors will follow this in a parallel fashion for the two copies: in the first step, there are only two source nodes to compute, and after this, the next node in this order will always already have a red pebble of the given processor on (at least) one of its predecessors. If the greedy algorithm is sophisticated enough, then this process requires no I/O, apart from possibly saving the $u_i$ to slow memory for future use. In any case, after $(2g+1) \cdot 2d$ computation steps (and a cost of $O(1)$), both processors have a red pebble on $u_{2d}$.

Assume that the first nodes of the main chains have incoming edges from the second input group (nodes $u_{d+1}, ..., u_{2d}$) as in Figure \ref{fig:gadget}. This means that after computing the input groups, the corresponding processors will have a red pebble on an in-neighbor $u_{2d}$ of the first main chain node (in both copies), so this is the node that is computed next by each processor. From here, the processors will continue processing the two main chains independently, since in each step, the next node in the main chain has an immediate predecessor (the previous node in the main chain) that already has a red pebble of the corresponding processor, and no red pebbles of the other processor. For $r=d+2$, even if the input groups are loaded from slow memory for each main chain node (which is the lowest-cost solution), this results in an I/O cost of $d$ for each node of the main chain, and a total cost of $\frac{n}{2} \cdot (1+d \cdot g) - O(1)$ altogether.

In contrast to this, if (after computing the input groups in the beginning) we use both processors to compute the main chain, only computing the two copies one after the other, then each node in the main chain only induces an I/O cost of $2g$. Since we cannot parallelize the two copies, the cost of the computation steps is now $n-O(1)$, but the total cost is still only $n \cdot (1+2g)-O(1)$. Ignoring the additive constant terms, this results in a ratio of
\[ \frac{1+ d _{\!} \cdot _{\!} g}{2+4 _{\!} \cdot _{\!} g} = \frac{1 + (\Delta_{in}-1) \cdot g}{2 + 4 _{\!} \cdot _{\!} g} \, . \]
For $g \geq 2$, we can lower bound this by $\frac{(\Delta_{in}-1) \cdot g}{5 _{\!} \cdot _{\!} g} \geq \frac{1}{5} _{\!} \cdot _{\!} (\Delta_{in}-1)$ which is again lower bounded by the slightly simpler expression $\frac{1}{5} _{\!} \cdot _{\!}\Delta_{in}-1$ used in the lemma.

\subsubsection{Lemma \ref{lem:greedy_lower}: lower bound of $\frac{2}{3} \cdot g$}

The second claim requires a construction where the optimal pebbling has I/O cost of $O(1)$ only, but the greedy algorithm commits the mistake of using $\Theta(n)$ I/O steps. Let $m=\frac{n}{6}$. Our DAG consists of $2$ chains of length $2m$ each, and $2$ chains of length $m$ each; let us denote the nodes in the chains by $u_{1}, ..., u_{2m}$, $v_1, ..., v_{2m}$, $w_1, ..., w_{m}$, $z_1, ..., z_m$. Apart from $(u_i, u_{i+1}), (v_i, v_{i+1})$ for all $i \in [2m-1]$, and $(w_i, w_{i+1}), (z_i, z_{i+1})$ for all $i \in [m-1]$, we add the edges $(u_i, w_i)$, $(v_{m+i}, w_i)$, $(u_{m+i}, z_i)$ and $(v_i, z_i)$ to the DAG for all $i \in [m]$. We also add the edges $(u_{2m}, z_1)$ and $(v_{2m}, w_1)$. We select $k=2$ and $r=\infty$, i.e., we never need to remove an already placed red pebble.

The optimal pebbling strategy here is to first spend $m$ compute steps to compute $u_{1}, ..., u_{m}$ on processor $p_1$, and $v_{1}, ..., v_{m}$ on processor $p_2$ simultaneously. We then switch the processors for the second halves of the chains, i.e.\ here we transfer $u_m$ to $p_2$ and $v_m$ to $p_1$ simultaneously, at a cost of $2g$. After this, we spend another $m$ compute steps to compute $u_{m+1}, ..., u_{2m}$ on $p_2$, and $v_{m+1}, ..., v_{2m}$ on $p_1$. After this, we can compute $w_{1}, ..., w_{m}$ on $p_1$, and $z_{1}, ..., z_{m}$ on $p_2$ in a parallel manner. The total cost of the pebbling is only $3m + 2g = \frac{n}{2} + O(1)$.

In contrast to this, the greedy algorithm first computes $u_{1}, ..., u_{2m}$ on processor $p_1$, and $v_1, ..., v_{2m}$ on processor $p_2$ (without changing processors in the middle of the chains). However, this implies that all the nodes $w_i$ and $z_i$ will require the communication of at least one data value to another processor, regardless of which processor they are computed on. Note that recomputing the other chain, e.g., the nodes of $v_1, ..., v_{2m}$ with $p_1$ is not an option under the definition of our greedy strategy, since the strategy specifies that the next node to compute on $p_1$ (which was not computed on $p_1$ before) must be one of the nodes the nodes $w_i$ or $z_i$. Altogether, this greedy strategy results in an I/O cost of $ m \cdot 2g$ at least, and hence a total cost of at least $\frac{n}{2} + \frac{n}{3} \cdot g$. As such, the ratio between the two solutions is asymptotically $\frac{3+2g}{3}=1+\frac{2}{3} \cdot g$.

\subsection{Transferring lower bounds from SPP to MPP}

\subsubsection{Proofs of Lemma \ref{lem:k_to_1} and Corollary \ref{coro:lb_pebbling_1}}

\renewcommand*{\proofname}{Proof of Lemma \ref{lem:k_to_1}}

As mentioned before, the key insight for Lemma \ref{lem:k_to_1} is that the parallel I/O steps in MPP can each be simulated with at most $k$ sequential I/O rules.

\begin{proof}
Let $G, r, k, L$ be as in the lemma. Towards a contradiction, suppose that there is an MPP pebbling strategy $P_m$ for $k$ processors that uses $\ell < L$ many I/O operations.
Transform strategy $P_m$ to an SPP strategy $P_s$ by replacing each parallel transition rule with up to $k$ sequential transition rules as follows:
Rule \ref{rule:multi-remove} can be directly translated to \ref{rule:single-remove}.
Any other rule $t \in \{ \text{\ref{rule:multi-bluebyred}, \ref{rule:multi-redbyblue}, \ref{rule:multi-compute}} \}$ that is based on a \colorful selection $f_m$ and vertices $v_1,v_2, \ldots, v_m$ can be replaced by $m$ corresponding sequential rules applied to $v_1,v_2, \ldots, v_m$.
It follows that $P_s$ requires at most $k \cdot \ell$ I/O operations; with $k \cdot \ell < k \cdot L$, we reach the desired contradiction. Hence, any multiprocessor pebbling strategy must use at least $L$ many I/O operations.
\end{proof}

Considering the result of Lemma \ref{lem:k_to_1}, and adding the lower bound on computation costs established in Lemma \ref{lem:trivial_bounds}, Corollary \ref{coro:lb_pebbling_1} already follows.

\renewcommand*{\proofname}{Proof of Corollary \ref{coro:lb_pebbling_1}}

\begin{proof}
Let $G, n, r, k, g, L$ be as in the lemma.
The previous lemma shows that any multiprocessor strategy requires $L/k$ many I/O operations, translating 
to at least $g \cdot L/k$ I/O costs. 
In the best case, a strategy does not recompute nodes and each processor computes $\frac{n}{k}$ nodes. It follows that the costs of any multiprocessor pebbling strategy is at least $g \cdot L/k + {n} /{k}$.
\end{proof}

\subsubsection{Tight lower bound: proof of Lemma \ref{lem:IO_lower_strict}}

To prove the lemma, we compose $k$ copies of the following gadget: we begin with a chain of length $2m$, and we add edges form the $i$-th to the $(m+i)$-th node, to ensure that the data of the first $m$ nodes are also required later in the chain. More formally, we consider a chain of nodes $u_1, ..., u_m, u_{m+1}, ..., u_{2m}$, with $(u_i, u_{i+1}) \in E$ for $i \in [2m-1]$, and $(u_i, u_{m+i}) \in E$ for $i \in [m]$. We take $k$ independent copies of this gadget, and set $r=3 \cdot k$. Note that $n=2 _{\!} \cdot _{\!} k _{\!} \cdot _{\!} m$.

In this DAG, the nodes $u_1, ..., u_m, u_{m+1}, ..., u_{2m}$ must be computed in this order. Furthermore, all of $u_1, ..., u_m$ are required later for computing $u_{m+1}, ..., u_{2m}$, respectively. This means that almost all these data values need to be saved into slow memory and loaded later, since we only have $O(1)$ red pebbles ($r=3 \cdot k$ in the SPP case, and $r=3$ in the MPP case).

In SPP, we need to save and load $m-(3 \cdot k-2)$ values to/from slow memory in each component of the DAG (we need $2$ red pebbles to compute through the chain, so $(r-2)$ red pebbles can be kept on some nodes, e.g., $u_1, ..., u_{r-2}$). This results in $2 _{\!} \cdot _{\!} (m-(3 _{\!} \cdot _{\!} k-2)) = 2 _{\!} \cdot _{\!} m - O(1)$ I/O operations per copy, and $L=k _{\!} \cdot _{\!} 2 _{\!} \cdot _{\!} m - O(1)=n - O(1)$ I/O operations altogether in SPP, since the components can be computed one after the other.

However, in MPP with $k$ processors, both the compute and the I/O operations can be parallelized over the $k$ copies. More specifically, with $k$ processors and only $r=3$, all the values $u_1, ..., u_m$ need to be saved to slow memory and loaded back later (in any pebbling strategy), resulting in $2 _{\!} \cdot _{\!} m$ I/O steps for each component. We can compute each of the copies on a separate processor in a parallel manner. This leads to an I/O cost of $2 _{\!} \cdot _{\!} g _{\!} \cdot _{\!} m = \frac{n}{k} _{\!} \cdot _{\!} g$, and a compute cost of $2 _{\!} \cdot _{\!} m=\frac{n}{k}$ for the pebbling. With $L=n - O(1)$ and $k, g \in O(1)$, we have $\frac{L}{k} _{\!} \cdot _{\!} g =\frac{n}{k} _{\!} \cdot _{\!} g -O(1)$, so for the pebbling of optimal cost, we indeed have
\[ \texttt{OPT} = \frac{n}{k} + \frac{n}{k} _{\!} \cdot _{\!} g = \frac{n}{k} + \frac{L}{k} _{\!} \cdot _{\!} g + O(1) \, . \]

\section{Impact of having more processors: proofs for Section \ref{sec:tradeoff}} \label{app:tradeoff}

We now provide the proofs omitted in Section \ref{sec:tradeoff}.

\subsection{Simpler proofs for the fair and practical comparisons}

Recall that in the fair case, the optimum cost can decrease by a factor $k$ at most.

\renewcommand*{\proofname}{Proof of Lemma \ref{lem:factor_k}}

\begin{proof}
As before in Lemma \ref{lem:k_to_1}, any pebbling strategy with $k$ processors can be transformed into a pebbling for the $k=1$ case, by keeping the red pebbles on the same (altogether at most $r_0$) nodes as in the parallel case, and splitting each parallel \ref{rule:multi-redbyblue}, \ref{rule:multi-bluebyred} or \ref{rule:multi-compute} step into at most $k$ separate single-processor steps. The cost of the resulting pebbling is at most a factor $k$ larger than the original cost. This implies $\texttt{OPT}^{(1)} \leq k \cdot \texttt{OPT}^{(k)}$.

A factor $k$ decrease is easily possible in a DAG with $k$ independent chains of length $\frac{n}{k}$: the computation cost here drops from $n$ to $\frac{n}{k}$, and no I/O steps are required in either case.
\end{proof}

As for the maximal cost increase in the fair case, we can first use a zipper gadget to show an increase of a $(2g+1)$ factor asymptotically.

\renewcommand*{\proofname}{Proof of Lemma \ref{lem:fair_increase}, for a factor $(2g+1)-o(1)$ only}

\begin{proof}
Consider zipper gadget with $r_0=2 \cdot d + 4$. With $k=1$ and $r=r_0$, we can easily keep both input groups and the two current nodes of the chain in fast memory, thus requiring no I/O, and resulting in a cost of $\texttt{OPT}^{(1)}=n$. However, with $k=2$ and $r=d+2$, the best strategy is to keep an input group in the fast memory of both processors, and compute the chain nodes alternatingly on the two processors. This requires two I/O steps for each node in the main chain, giving $\texttt{OPT}^{(2)}=n+2 g \cdot n-O(1)$, which is asymptotically a factor $(2g+1)$ increase.
\end{proof}

With further technical steps, one can reorganize this into a construction with a pool of $k \cdot d$ source nodes, and a main chain where each node has $d$ inputs from this pool, but the inputs of consecutive chain nodes only have a small intersection. This forces any pebbling strategy to load almost all (a $\frac{k-1}{k}$ fraction) the inputs of every chain node from slow memory, giving the factor shown in the original lemma. Since this requires a more detailed analysis, we present it separately below in Appendix \ref{app:fair_increase}.

We can also show that in the fair case, the optimum can be non-monotonic in $k$.

\renewcommand*{\proofname}{Proof of Lemma \ref{lem:nonmonotone}}

\begin{proof}
Consider two copies of the zipper gadget with $r_0=4 \cdot d + 8$. With $k=1$ processor and $r=r_0$, the optimal cost is $n$. With $k=2$ and $r=2 \cdot d + 4$, we can execute the two copies of the DAG on two processors, as the two input groups and two chain nodes still all fit into fast memory; hence no I/O is needed, and $\texttt{OPT}^{(2)}=\frac{n}{2}$. For $k=4$ and $r= d + 2$, we again have to store the $4$ input groups on $4$ distinct processor, once again giving a cost of $2 \! \cdot \! g+1$ per chain node, and in total $\texttt{OPT}^{(4)}=\frac{n}{2}+2 \! \cdot \! g \cdot \frac{n}{2}-O(1)$.
\end{proof}

Finally, turning to the practical case, the zipper gadget also demonstrates that we can easily achieve a superlinear speedup.

\renewcommand*{\proofname}{Proof of Lemma \ref{lem:suplin}}

\begin{proof}
Consider the zipper gadget with $r_0= d + 2$. With $k=1$ processor, the optimal pebbling has $n$ computation steps, and $d \cdot n - O(1)$ steps of I/O: for each chain node, we need to load $d$ values from slow memory. On the other hand, with $k=2$, we can keep a separate input group in the fast memory of both processors, compute the chain nodes alternatingly, and only communicate the chain nodes between the processors. This results in $n - O(1)$ compute steps, and only $2 \cdot n - O(1)$ I/O steps ($2$ for each node of the main chain). This gives a ratio of $\frac{(g \cdot d+1) \cdot n - O(1)}{(2g+1) \cdot n - O(1)}$ for costs. Asymptotically (for $n$ large enough), and using $\Delta_{in}= d+1$, this is a ratio of $\frac{g \cdot (\Delta_{in}-1)+1}{2g+1}$. By selecting a large enough constant value for $g$, the claim follows for any $\varepsilon>0$.
\end{proof}

\subsection{For the tight bound in Lemma \ref{lem:fair_increase}} \label{app:fair_increase}

We now show how to achieve a cost increase of essentially a $\frac{k-1}{k} \cdot \Delta_{in} \cdot g$ factor in Lemma \ref{lem:fair_increase}. We first describe the construction for the simplest case of $k=2$, and then we show how to generalize it to arbitrarily large $k$.

For $k=2$, instead of having $2$ input groups of size $d$, our construction will have $4$ input ``subgroups'', named $S_{1,1}$, $S_{1,2}$, $S_{2,1}$ and $S_{2,1}$. Each of these will consist of $\frac{d}{2}$ nodes (assume for simplicity that $d$ is an even number). The main chain nodes will now periodically alternate between $4$ different kind of nodes: they will have incoming edges from all of the nodes in (i) $S_{1,1}$ and $S_{1,2}$, (ii) $S_{2,1}$ and $S_{2,2}$, (iii) $S_{1,1}$ and $S_{2,2}$, (iv) $S_{1,2}$ and $S_{2,1}$. As before, each main chain node also has an edge to the following node; as such, the in-degree of nodes in the main chain is $d+1$ as before, and hence we can still have $r_0=2d+4$.

The best pebbling strategy in this DAG is to compute the first two main chain nodes (in each cycle of $4$) by processors $p_1$ and $p_2$, and then the next two nodes again by two different processors, e.g., the 3rd node by $p_1$ and the 4th node by $p_2$ (or vice versa). In this case, the 3rd and 4th nodes both require us to load $\frac{d}{2}$ nodes from slow memory. After this, the 1st and 2nd nodes of the next cycle of $4$ will also require us to load $\frac{d}{2}$ nodes from slow memory. Thus apart from the first $2$, every main chain node induces an I/O cost of $\frac{d}{2} \cdot g$. Note that the node in the subgroups only need to be saved into slow memory once, so this can be considered part of the $O(1)$ initialization cost. The current main chain node also needs to be communicated to the other processor sometimes, incurring an I/O cost of $2g$, but we ignore this for our formula, since it becomes irrelevant besides $d$ for $d$ large enough.

As such, the total cost of this pebbling strategy is at least $\frac{d}{2} \cdot n - O(1) = \frac{\Delta_{in}-1}{2} \cdot n - O(1)$. One can check that this is indeed the best pebbling strategy in our construction. In particular, the intersection between the in-neighbors of any two nodes in the cycle of $4$ is always at most $\frac{d}{2}$. Hence whenever a processor $p_i$ computes the next node of the main chain, and this has a role in its $4$-cycle (i.e., different index modulo $4$) than the last node computed by $p_i$, then this requires us to load at least $\frac{d}{2}$ data values into the fast memory of $p_i$. The only reasonable alternative is to select a processor, e.g., $p_1$, and keep a specific pair of subgroups, e.g., $S_{1,1}$ and $S_{1,2}$, always on $p_1$ for the entire cycle of $4$; this ensures that the next chain node with inputs $S_{1,1}$ and $S_{1,2}$ requires no I/O at all (except for the steps to communicate the previous main chain node). However, this implies that the other $3$ main chain nodes in the cycle all need to be computed by $p_2$; throughout an entire cycle, this implies that $p_2$ has to load $\frac{d}{2}+\frac{d}{2}+d = 2d$ values, again averaging to $\frac{d}{2}$ I/O operations per chain node.

Let us now consider a general $k$ value. Let $r_0=k \cdot (d+2)$. Assume for simplicity that $k$ is a prime number; this still allows us to increase $k$ arbitrarily large. We will now form $k^2$ different subgroups, each of size $\frac{d}{k}$, labeled with $S_{i, j}$ with $i, j \in \{ 1, ..., k \}$. The main chain nodes again periodically follow a cycle of length $k^2$: for all $i, j \in \{ 1, ..., k \}$, the $((i-1) \cdot k + j)$-th node in the cycle will have an incoming edge from each of the groups $S_{(j+(\ell-1) \cdot i) \! \mod k, \, \ell}$ for $\ell \in \{1, ..., k\}$. For instance, the cycle for $k=3$ will be (i) $S_{1,1}$, $S_{1,2}$, $S_{1,3}$ (ii) $S_{2,1}$, $S_{2,2}$, $S_{2,3}$, (iii) $S_{3,1}$, $S_{3,2}$, $S_{3,3}$, (iv) $S_{1,1}$, $S_{2,2}$, $S_{3,3}$, (v) $S_{2,1}$, $S_{3,2}$, $S_{1,3}$, (vi) $S_{3,1}$, $S_{1,2}$, $S_{2,3}$, (vii) $S_{1,1}$, $S_{3,2}$, $S_{2,3}$, (viii) $S_{2,1}$, $S_{1,2}$, $S_{3,3}$, (ix) $S_{3,1}$, $S_{2,2}$,  $S_{1,3}$. Note that any two of these $k$-tuples of subgroups only intersect in $0$ or $1$ subgroups; as such, whenever a processor computes two different nodes of the cycle consecutively, this requires loading $\frac{k-1}{k} \cdot d$ data values to slow memory at least (or up to $d$ values, if the two $k$-tuples are disjoint). Also, for simplicity, let us call the main chain nodes $1$ to $k$, main chain nodes $(k+1)$ to $2k$, etc., \textit{batches} of nodes; our entire cycle consists of $k$ batches, each batch contains $k$ distinct $k$-tuples, and any two $k$-tuples within a batch are always entirely disjoint.

In this construction, we can still follow the same pebbling strategy as in the $k=2$ case before: compute each of the first batch of main chain nodes on the $k$ different processors, then again distribute the next batch of $k$ main chain nodes over all the $k$ processors, and so on. In this strategy, we need at least $\frac{k-1}{k} \cdot d$ I/O operations for each main chain node (except for possibly the first cycle).

Alternatively, we can again select one or more processors that keep the same $k$-tuple of subgroups in fast memory for an entire cycle; however, in this case, the remaining processors need to compute all the other nodes in the cycle, which results in a similarly high cost. This is because the $k$-tuples within a batch have no intersection at all, so distributing them among less than $k$ processors implies that we need to load $d$ values instead of $\frac{k-1}{k} \cdot d$ on some occasions. In particular, if e.g., there is only one processor $p_1$ that keeps the same $k$-tuple for a cycle (and hence it requires no I/O steps, apart from loading the preceding chain node), then the remaining $(k-1)$ processors need to compute the other $(k^2-1)$ nodes in the cycle. In $(k-1)$ of the batches (i.e., except for the single batch that has the node computed on $p_1$), this means that at least one processor needs to compute at least $2$ chain nodes of the batch; these nodes have entirely disjoint inputs, so this leads to an extra cost of $d-\frac{k-1}{k} \cdot d = \frac{d}{k}$ per batch, and $\frac{k-1}{k} \cdot d$ over the $(k-1)$ batches, which is exactly the I/O cost saved by $p_1$. In general, if we have $k'$ processors that always keep the same $k$-tuple in fast memory for an entire cycle, then the remaining $(k-k')$ processors must compute the other $(k^2-k')$ chain nodes; a simple analysis shows that this result in altogether at least $(k-1) \cdot k'$ occasions when two nodes within a batch are consecutively computed on the same processor. This comes with an extra I/O cost of $(k-1) \cdot k' \cdot \frac{d}{k} = k' \cdot \frac{k-1}{k} \cdot d$, which is again exactly the I/O cost saved by the $k'$ processors. As such, other pebbling strategies also result in at least the same cost.

As such, the pebbling strategy has a total cost of
\[ n + \frac{k-1}{k} \cdot d \cdot g \cdot n - O(1) = \left(\frac{k-1}{k} \cdot (\Delta_{in}-1) \cdot g + 1 \right) \cdot n - O(1) \, .\]
In contrast to this, in the single-processor case of $k=1$, all the subgroups still fit into fast memory simultaneously, so there is no need for I/O at all; the optimal pebbling cost is simply $n$. Asymptotically, the ratio between the two values is 
\[ \frac{k-1}{k} \cdot (\Delta_{in}-1) \cdot g + 1 \, .\]
This essentially matches the upper bound for $k$ large enough.

\subsection{Number of I/O steps in the optimum}

Recall that $\texttt{OPT}_{I/O}$ was defined as the number of I/O steps in the optimal solution. In case there are multiple minimum-cost pebblings with a different number of I/O steps, we define $\texttt{OPT}_{I/O}$ to be the minimum of these values, to make the definition precise. However, since our constructions below only allow a single optimal strategy, our claims also hold for any other tie-breaking rule among minimum-cost solutions.

It is easy to see that $\texttt{OPT}_{I/O}$ can notably increase when adding multiple processors, both in the fair and in the practical setting. For a simple example, consider two chain DAGs of length $2g+1$ each, and add a common source node $v_s$ that has an outgoing edge to the first node of both chains, and a common target node $v_t$ that has an incoming edge from the last node of both chains. Let $r=\infty$ for simplicity. With $k=1$, this DAG is computed on a single processor, without requiring any I/O. However, with $k=2$, it is more beneficial to parallelize the two chains on the two processors (this adds an I/O cost of $2g$, but saves $2g+1$ compute steps). As such, the number of I/O steps increases from $0$ to $2$.

We can then take many, up to $\Theta(n)$ copies of this small DAG, and merge them sequentially, i.e., always merge the node $v_t$ of the current copy and the node $v_s$ of the next copy into a single node. All of the copies will still behave as before; hence the I/O cost for $k\!=\!1$ will be $0$, but the I/O cost for $k\!=\!2$ will be $2 \! \cdot \! g \! \cdot \! \Theta(n)$. This shows a multiplicative increase of an infinite factor, or an additive increase of a $\Theta(n)$ term.

More surprisingly, $\texttt{OPT}_{I/O}$ can also decrease in a similar way. For this, we consider a DAG consisting of two connected components. The first one is simply a chain DAG of length $m$. The second one is variant of the zipper gadget, with the modification that instead of $2$, it now has $4$ input groups of size $d$ (still with chains of size $2g$ added to discourage recomputation), and the main chain is alternating periodically between these $4$ input groups (besides having the previous main chain node as a predecessor, as before). Let $r_0=2d+4$, so that we still have $\frac{r_0}{2} \geq d+2=\Delta_{in}+1$ for the fair case. Let the main chain in this second component consist of $n_0=\frac{m}{d\cdot (2g+1)+1}$ nodes; note that this still implies $n_0=\Theta(n)$.

With $k=1$, the optimal pebbling here is to compute the chain component (a compute cost of $m$), and then in the second component, to always use I/O steps to reload the nodes in the input groups, since recomputing them is more costly. In both cases, with $k=1$, we can only essentially store $2$ input groups in fast memory at the same time (plus $2$ more nodes), so each cycle of $4$ main chain nodes will incur at least $(2d-2)$ I/O steps (in fact, almost $3d$ even with the best strategy). This implies that the optimal pebbling strategy contains $(2d-2) \cdot \frac{n_0}{4} = \Theta(n)$ I/O steps.

On the other hand, if $k=2$, then one of the processors (say, $p_1$) has to compute the chain component alone, so the number of compute steps is at least $m$. However, in this case, for a smaller total cost, $p_2$ is actually better off using recomputation for each main chain node in the other component: the recomputation of an input group requires $d\cdot (2g+1)$ compute steps, and then a further step to compute the next node in the main chain. That is, in $(d\cdot (2g+1)+1) \cdot n_0 = m$ steps, $p_2$ is actually able to compute the entire second component without ever requiring an I/O step, and instead recomputing the input group nodes every time. Since $p_1$ requires $m$ compute steps anyway, this does not increase the number of computation steps, and hence provides an optimal pebbling strategy of cost $m$, without any I/O steps at all.

\section{APX-hardness: proof of Theorem \ref{th:apx}} \label{app:apx}

As outlined before, the proof is based on the reduction idea in \cite{RBpebbling3}; however, it requires significant adjustments to work in a setting with computation costs. Note that SPP with computation costs, as defined in \cite{RBpebbling3}, assigns a cost of $1$ to steps (R1-S) and (R2-S), and a cost of some constant $\epsilon>0$ to (R3-S). Our MPP with $k=1$ (assigning a cost of $g$ to \ref{rule:multi-redbyblue}, \ref{rule:multi-bluebyred} and a cost of $1$ to \ref{rule:multi-compute}) is indeed equivalent to this with the appropriate scaling factor. As such, we present the proof using our notation for MPP, with $k=1$. Note that the proof also holds without changes for the SPP definition variant where pebbles are replaced in the I/O rules.

The base of the reduction is the vertex cover problem: given an undirected graph $G'(V',E')$ on $N$ nodes, the goal is to find the smallest set of nodes $V_0 \subseteq V'$ such that every edge $e' \in E'$ is incident to at least one node in $V_0$, where $|V_0|$ is minimal. The original proof in \cite{RBpebbling3} uses the property that assuming the Unique Games Conjecture, it is not possible to approximate vertex cover to any multiplicative factor below $2$. Instead, we use the fact that the problem is APX-complete already in $3$-regular graphs \cite{vcapprox, MIS_bounded2}.

Furthermore, note that a $3$-regular graph on $N$ nodes has $M=\frac{3}{2} _{\!} \cdot _{\!} N$ edges, and each node can cover at most $3$ edges; as such, the smallest vertex cover still has size at least $\frac{N}{2}$.

\subsection{Main idea of the original reduction}

On a high level, the idea of the construction in \cite{RBpebbling3} is as follows. The DAG consists of several gadgets called \textit{groups}. Each group consists of exactly $(r-1)$ nodes, and has one or more so-called \textit{target nodes}, i.e., a node that has an incoming edge from all the $(r-1)$ nodes in the group. This allows for a simple analysis, since the computation of each target node requires all the $r$ red pebbles: $(r-1)$ on the group, and one on the target node. The pebbling strategy is then essentially described by the order of computing the target nodes.

The construction has $2 \cdot N$ groups: two groups corresponding to each node $v' \in V'$ of the input graph in the vertex cover problem, called the \textit{first} and \textit{second} group of the node; we will use the shorthand notation $S_1(v')$ and $S_2(v')$ for these two groups, respectively. For each node $v' \in V'$, $S_1(v')$ has one of its target nodes included in the second group $S_2(u')$ of every node $u'$ that is adjacent to $v'$ in $G'$. The group $S_2(v')$ only has a single target node. Finally, $S_1(v')$ and $S_2(v')$ intersect in a large number of nodes; in fact, they are almost identical (the number of nodes differing between them becomes irrelevant asymptotically).

In this construction, the optimal strategy to minimize I/O costs is to visit the first and second group of a node consecutively whenever possible; otherwise, the intersection of the two groups has to be saved to slow memory and then reloaded later. In fact, the sizes of the gadgets are chosen such that asymptotically, all other I/O costs become irrelevant, and the total I/O cost is proportional to the number of nodes $v' \in V'$ where the visits to $S_1(v')$ and $S_2(v')$ is non-consecutive. However, recall that if $(u', v') \in E'$, then $S_2(u')$ contains a target node of $S_1(v')$, and $S_2(v')$ contains a target node of $S_1(u')$. As such, if we visit $S_1(v')$ and $S_2(v')$ consecutively at some point, then this means that $S_1(u')$ must be visited before, and $S_2(u')$ must be visited after this point; hence the visits to the groups of $u'$ cannot be consecutive. More specifically, the set of nodes $v' \in V'$ visited consecutively can only form an independent set in $G'$. As such, the optimal strategy is to visit the first groups for the smallest vertex cover, then both groups of the complementing (largest) independent set consecutively, and then the second groups of the smallest vertex cover. This way, the (asymptotic) I/O cost becomes proportional to the size of the smallest vertex cover in $G'$, and hence inapproximbaility results for vertex cover also carry over to pebbling.

\subsection{Main differences for computation costs}

To adapt this construction to our setting, the following main modifications are required:
\begin{enumerate}[itemsep=-10pt]
 \item The original proof considers SPP in the one-shot model, where recomputation is explicitly banned; however, in standard SPP or MPP, recomputation is allowed. Due to this, we need to add simple gadgets to the groups of the construction to ensure that recomputation is not beneficial in any case, and hence we can apply the same analysis as before.
 \item Most importantly, if computational steps also have a cost (as in MPP, or in SPP with computation costs), then this adds a further cost of $n$ to any pebbling strategy. Note that this total computation cost is indeed exactly $n$ if $k=1$ and no recomputation happens. Since the number of I/O steps happening in the construction is originally $o(n)$, this makes the I/O costs asymptotically negligible, and hence the original proof strategy does not apply: finding a much smaller vertex cover has no asymptotic effect on the total cost. As such, the most important adjustment is to reduce the size of the gadgets in the construction such that the new construction will have $n=\Theta(N)$, or in other words, each group gadget has constant size. In this case, a difference of $\Theta(N)$ in I/O costs also results in a difference of $\Theta(n)$ in total costs. 
 \item Since the number of target nodes required in the construction is $2 _{\!} \cdot _{\!} M$, we restrict ourselves to vertex cover in $3$-regular graphs, to ensure that the number of target nodes is $O(N)$, and hence we can have $n=\Theta(N)$ indeed.
 \item Finally, the size reduction of group gadgets to $O(1)$ creates a new problem that needs to be addressed. In particular, $(u', v') \in E'$, then visiting $S_2(u')$ after $S_1(v')$ can also save a small I/O cost, since the target node of $S_1(v')$ in $S_2(u')$ already has a red pebble on it, as it was just computed. This implies that I/O costs may not be proportional to the size of a vertex cover: we can also save I/O by visiting the groups of adjacent nodes consecutively. In the original construction this could be ignored, since the target nodes were asymptotically negligible compared to the size of group gadgets. To resolve this, we slightly modify the design of group gadgets and target nodes, ensuring that the last target node computed for a group gadget is never directly included in another group, and hence the same method does not save any I/O costs.
\end{enumerate}
We first discuss point (1) separately, since this requires a simple gadget that is independent from the rest of the construction. We then describe the modified construction and analyze its properties.

\subsection{Disabling recomputation}

In order to ensure that recomputation is not beneficial in or DAG, we use the same trick as in the zipper gadget: in front of each original source node $u$ in the DAG, we attach a chain of length $2g$, i.e. nodes $u_1, ..., u_{2g}$ such that $(u_i, u_{i+1}) \in E$ for all $i \in [2g-1]$, and $(u_{2g}, u) \in E$.

In any reasonable pebbling, this chain can be computed consecutively: first computing $u_1$, then computing $u_2$, deleting $u_1$, computing $u_3$, deleting $u_2$, and so on. This only requires $2$ red pebbles at any point, so it still remains possible to compute any group gadget: even when computing the last node in the group, we still must have $2$ free red pebbles, one to place on this last node, and one to place on the target node to compute. Furthermore, the gadget ensures that if there are no red pebbles on the chain, then recomputing $u$ results in a cost of $2g+1$ (i.e., $2g+1$ new compute steps), whereas saving it to slow memory earlier and then reloading it only has a cost of $2g$, hence recomputation is never beneficial.

More formally, any pebbling can be transformed (without increasing its cost) into a strategy that computes the nodes of such a chain consecutively, and does not execute any I/O steps on $u_1, ..., u_{2g}$. Consider the step when $u_2$ is computed, and hence there is a red pebble on both $u_1$ and $u_2$. We can move forward the computation of the whole chain (as described above, using only $2$ red pebbles) to immediately after this point, and hence compute $u$ already here. If a red pebble was kept on any of the chain nodes after this point, we can keep it on $u$ instead; otherwise, if any of $u_1, ..., u_{2g}$ was saved to slow memory or reloaded from slow memory afterwards, we can save/reload $u$ instead.

Note that the reason why we only apply the above modification to source nodes is that non-source nodes will have a very limited role in our construction, and it is easy to show that it is not worth recomputing any of these anyway.

This modification allows us to disregard the recomputation question in the rest of the analysis, since it is independent from the remaining properties of the construction. That is, we will present the construction without mentioning these gadgets, but assuming that source nodes cannot be recomputed in a reasonable solution. The only effect of this modification for our cost is that it increases the number of nodes in the DAG by a $2 \cdot g$ term for every original source node, i.e., by at most a $(2 g + 1)$ factor. Since we select $g$ to be a constant, this is only a constant factor increase in the DAG size, and hence it will not change the fact that the I/O costs in the construction are linear in $n$.

\subsection{Our construction}

Our modified DAG construction has parameters $B_0$ and $B_1$, with $B_0, B_1 \in O(1)$.

As before, each original node $v' \in V'$ will have group gadgets $S_1(v')$ and $S_2(v')$. All these $2N$ groups gadgets will essentially have a size of $B_0 + B_1 + 1$ (the first groups are technically a bit more complicated), and we set $r=B_0+B_1+2$. Every target node will have exactly $B_0+B_1+1$ in-neighbors, so the computation of any target node always requires all the available red pebbles, and hence there can be no red pebbles anywhere else in the DAG.

For each $v' \in V'$, $S_1(v')$ and $S_2(v')$ intersect in $B_0$ so-called common nodes, to motivate the consecutive visitation of the two groups. Furthermore, for each first group $S_1(v')$, we add $B_1$ nodes only present in this group, to motivate the consecutive computation of all the target nodes of $S_1(v')$ when it is visited. In $S_2(v')$, we include a target node of $S_1(u')$ for all $3$ nodes $u'$ adjacent to $v'$, and add $(B_1-2)$ further extra nodes to ensure that the group has size $B_0+B_1+1$. We select $B_1 > 3$ for our parameters; this also ensures $(B_1-2) > 0$.

As for the target nodes: each group $S_2(v')$ has a single target node. On the other hand, $S_1(v')$ has $3$ distinct targets, one for each neighbor $u'$ of $v'$ in $G'$. However, instead of having a single target node for $u'$, we now add a chain of $4$ nodes $w_0, w_1, w_2, w_3$ to our DAG to establish this relation, with $(w_0, w_1), (w_1, w_2), (w_2, w_3) \in E$. From these nodes, $w_0$ is a source of our DAG (can be considered as a ``pseudo-member'' of $S_1(v')$); $w_1$ is the target node included in $S_2(u')$; and each of $w_1$, $w_2$ and $w_3$ has an incoming edge from all nodes in $S_1(v')$. This ensures the following properties:
\begin{itemize}
 \item For computing $w_i$ for any $i \in \{1, 2, 3\}$, we still require all the red pebbles: $B_0+B_1$ pebbles on the rest of $S_1(v')$, one red pebble on $w_{i-1}$, and one red pebble on the $w_i$ being computed.
 \item We still need to visit $S_1(v')$ before $S_2(u')$, since $w_1$ is included in $S_2(u')$.
 \item If all target nodes of $S_1(v')$ are computed when visiting this group, then there is no red pebble on $w_1$ at the end of this process. This is because $w_2$ and $w_3$ has to be computed after $w_1$, and when computing $w_3$, we have no free red pebble to leave on $w_1$. In other words, in any reasonable pebbling, $w_1$ needs to be saved to slow memory while $w_2$ and $w_3$ is computed, and reloaded later. This ensures that we cannot save any cost by visiting $S_2(u')$ immediately after $S_1(v')$.
\end{itemize}

Intuitively, when visiting $S_1(v')$, the reasonable pebbling strategy is always to (i) keep $B_0+B_1$ pebbles on the first group, and use the remaining $2$ red pebbles to (ii) go over each target chain of $S_1(v')$, and (iii) compute $w_0, w_1, w_2, w_3$ in this order.

Note that except for the target nodes of second groups and $w_1, w_2, w_3$ in the target chains, all the other nodes are source nodes; each of these is replaced in the end by the chain gadget described before to discourage recomputation. On the other hand, the sink nodes in the DAG are (i) the final nodes $w_3$ in each target chain, which is $2 _{\!} \cdot _{\!} M$ nodes in total, and (ii) the $N$ target nodes of second groups. 

Altogether, we have $N _{\!} \cdot _{\!} (B_0+B_1)$ source nodes in first groups, $2 _{\!} \cdot _{\!} M = 3 _{\!} \cdot _{\!} N$ target chains of length $4$ (with a single source node), $N _{\!} \cdot _{\!} (B_1-2)$ source nodes in second groups, and $N$ target nodes for second groups. Recall that we attach a chain of length $2g$ to all sources. With this, the total number of nodes is
\begin{gather*}
 n = N  _{\!} \cdot _{\!} (B_0+B_1) \cdot (2g+1) + 3 _{\!} \cdot _{\!} N _{\!} \cdot _{\!} (2g + 4) \, +\,  N \cdot (B_1-2) \cdot (2g+1) + N =  \\
 = N \cdot (2gB_0 + B_0 + 4gB_1 + 2B_1 + 2g + 11) \, .
\end{gather*}
Since we have $B_0, B_1, g \in O(1)$, this implies $n \in O(N)$.

\subsection{Analysis and reduction}

We begin the analysis with some simple observations:
\begin{itemize}
 \item As discussed before, the computation of all target nodes and target chains require all the red pebbles.
 \item As also discussed, a reasonable pebbling strategy never uses recomputation. As such, the total cost of the computation steps is exactly $n$.
 \item The sink nodes play no significant role: the computation of the last sink node will require all previous sinks to be already saved to slow memory anyway, so we might as well save each sink node (except the last one to be computed) to slow memory as soon as it is computed, and then delete the red pebble from it immediately. These operations will add an inherent I/O cost of $(2 \cdot M + N -1) \cdot g$ to any pebbling.
 \item The source nodes in second groups also play no significant role: they can be computed immediately before the target node of the second group is computed, and then deleted immediately afterwards.
\end{itemize}

As the next step, we observe that a reasonable pebbling strategy computes all the target chains of any first group $S_1(v')$ consecutively. Indeed, if we compute any other target node (of a first or second group) between computing two targets of $S_1(v')$, then this implies that the $B_1$ nodes that only appear in $S_1(v')$ need to be all saved to slow memory and then loaded back, at a cost of $2 \cdot B_1 \cdot g$. On the other hand, among the target nodes of $S_1(v')$, there are only $3$ nodes $w_1$ that are required for a future computation, and each of these is only needed a single time in the future (since their out-degree is $1$). As such, if we have $B_1 > 3$, then transferring the $B_1$ nodes in $S_1(v')$ to/from slow memory is actually more expensive than computing all target nodes of $S_1(v')$ consecutively, saving all the corresponding $w_1$ to slow memory, and loading them later at a total cost of at most $2 \cdot 3 \cdot g$. Hence if a pebbling strategy does not compute all target chains of $S_1(v')$ consecutively, then we can modify this to a strategy where these computations are consecutive, and it has strictly lower cost.

Note that there is a natural order to computing the target nodes of any first group: we consider the target chains in any desired order, and compute the chain from $w_0$ to $w_3$ in order, saving $w_1$ to slow memory (and also $w_3$ since it is a sink node), and deleting $w_0$ and $w_2$ when not needed anymore. This order does not induce any I/O costs besides the unavoidable ones mentioned so far. Note that the last target node to compute is a $w_3$ in one of the target chains, so as mentioned before, there can be no red pebbles on any of the $w_1$ when the computation of all target chains is finished.

This shows that every group gadget is indeed ``visited'' only once (when its targets are computed), and hence any pebbling strategy is essentially characterized by the order of visiting the groups. Besides the fixed, unavoidable I/O costs mentioned so far (total of $(2 \cdot M + N -1) \cdot g$ for saving sink nodes, and total of $2M \cdot 2 \cdot g$ for saving and loading the nodes $w_1$ in target chains), the only I/O operations in such a pebbling come from the fact that if $S_1(v')$ and $S_2(v')$ are not visited consecutively, then their $B_0$ common nodes need to be saved to and later loaded from slow memory, at a cost of $2 \cdot g \cdot  B_0$ for each such $v'$.

As in the original construction, the nodes $v'$ for which $S_1(v')$ and $S_2(v')$ are visited consecutively must form an independent set in $G'$, since for all $(u', v') \in E'$, $S_2(u')$ must be preceded by $S_1(v')$, and $S_2(v')$ must be preceded by $S_1(u')$. In other words, the nodes $v'$ where the groups are non-consecutive must form a vertex cover in $G'$. Let us denote the size of this vertex cover in $G'$ by $VC$, and the size of the smallest vertex cover in $G'$ by $\texttt{OPT}_{VC}$. Then the corresponding I/O cost in the pebbling strategy is $2  \cdot g \cdot B_0 \cdot VC$, and the total cost of the pebbling is
\[ n + (6 _{\!} \cdot _{\!} M + N -1) _{\!} \cdot _{\!} g + 2 _{\!} \cdot _{\!} g _{\!} \cdot B_{0\!} \cdot _{\!} VC \, . \]

This shows that finding an optimal pebbling is clearly equivalent to finding a minimum vertex cover in $G'$. For convenience, we use $n'$ to denote $n' := n + (6 \cdot M + N -1) \cdot g$, and $\alpha:=2 _{\!} \cdot _{\!} g _{\!} \cdot _{\!} B_{0}$. Recall that $n=O(N)$ so $n'=O(N)$, while $\alpha =O(1)$.

The above observations provide a natural L-reduction between the two problems. In particular, with $\texttt{OPT}_{VC} \leq N$, the cost of the optimal pebbling is at most $n'+\alpha _{\!} \cdot _{\!} N = O(N)$, whereas $\texttt{OPT}_{VC}\geq \frac{N}{2}$, so the two optimum values are indeed only a constant factor apart. Furthermore, the optimal pebbling strategy has a cost of $n' + \alpha _{\!} \cdot _{\!} \texttt{OPT}_{VC}$. Given any pebbling strategy (of cost $\texttt{SOL}$) in our construction, we can transform it into a reasonable schedule outlined above: we replace any recomputations by I/O steps, we compute all target nodes of a group when a group is first visited (in their natural order), and we remove any further steps that are unnecessary. According to our observations above, this transformation can only decrease the cost of our pebbling. We then consider the vertex cover (of size $VC$) defined by the group-pairs that are visited non-consecutively; if this has size $VC$, then our pebbling incurs a total cost of n' + $\alpha _{\!} \cdot _{\!} VC$. As such, we get that
\[ \texttt{SOL} - (n' + \alpha _{\!} \cdot _{\!} \texttt{OPT}_{VC}) \geq (n' + \alpha _{\!} \cdot _{\!} VC) - (n' + \alpha _{\!} \cdot _{\!} \texttt{OPT}_{VC}) = \alpha _{\!} \cdot _{\!} (VC - \texttt{OPT}_{VC}) \, , \]
which completes the L-reduction. Since Vertex-Cover is APX-hard, this shows that MPP (or SPP with computation costs) is also APX-hard, allowing no PTAS.

\section{Inapproximability of I/O costs: proof of Theorem \ref{th:inapprox}} \label{app:inapprox}

The main components of the proof of Theorem \ref{th:inapprox} have already been outlined in Section \ref{sec:inapx}. in this section, we discuss the technical details. Recall that all pebbles in our proof are red pebbles, since we cannot use any I/O steps.

We again point out that the main technical contribution of this theorem shows that it is NP-hard to decide whether we can a pebble a DAG using only (a given number of) red pebbles. This is very similar to the long known claim that the black pebble game is NP-hard \cite{blackNPhard}; as such, in some sense, this result is not entirely novel. However, \cite{blackNPhard} considers a slightly different version of pebbling where upon computing a node $v$, you also have an alternative option to \emph{slide} a pebble to $v$ from one of its in-neighbors. While it might be possible also to adapt this proof of \cite{blackNPhard} to the non-sliding case with some further gadgets and analysis, we present an entirely different NP-hardness proof instead. We believe that our reduction is somewhat simpler than that of \cite{blackNPhard} (which is based on a SAT problem variant), and the concept of level gadgets in our proof might be of independent interest, possibly being applicable in future result on pebbling problems. Finally, to our knowledge, it has not been explicitly observed in previous works that such an NP-hardness result leads to the given additive/multiplicative inapproximability bounds for SPP/MPP, even though this only requires some simple technical steps.

\subsection{Level gadgets}

We first formalize the properties of level gadgets. Given two consecutive levels, we denote the nodes on the earlier level by $u_1, \ldots,u_{\ell}$, and those in the later level by $v_1, \ldots, v_{\ell'}$. We always have $(u_i, u_{i+1}) \in E$ for $i \in [\ell-1]$ and $(v_i, v_{i+1}) \in E$ for $i \in [\ell'-1]$. We also have $(u_i, v_i) \in E$ for all $i \in [1, min(\ell, \ell')]$, and we have $(u_{\ell}, v_1) \in E$. Finally, if $\ell > \ell'$, then we have $(u_i, v_{\ell'}) \in E$ for all $(\ell'+1) \leq i \leq \ell$. Also, recall that for any level $u_1, \ldots,u_{\ell}$, any external node (i.e., which is not part of the preceding or subsequent level of the same tower) must have a directed edge either to none or to all of the nodes in $u_1, \ldots,u_{\ell}$, and any external node has an incoming edge either from none or from all of the nodes in $u_1, \ldots,u_{\ell}$.

Assume that all the nodes $u_1, \ldots, u_{\ell}$ have a pebble on them at a given point in time. Then with the use of a single extra pebble, we can proceed to the next level in such a gadget. More specifically, if $\ell=\ell'$, then we can first compute $v_1$, then delete $u_1$, then compute $v_2$, then delete $u_2$, and so, finally compute $v_{\ell'}$ and delete $u_{\ell}$, never having more than $(\ell+1)$ red pebbles on the gadget in the process. In case if $\ell' > \ell$, we can continue this process in the end by computing $v_{\ell+1}, \ldots, v_{\ell'}$ in this order; this again does not require more than $\ell'$ red pebbles at any point. If $\ell' < \ell$, then after computing $v_{\ell'}$, we can delete $u_{\ell'+1}, \ldots, u_{\ell}$ in any order, again not having more than $(\ell+1)$ red pebbles at any point.

Moreover, this is indeed the reasonable way to proceed to the next level in any pebbling. Formally, we consider two different cases. On the one hand, if all of the nodes $u_1, \ldots, u_{\ell}$ have a red pebble in the step when $v_{\ell'}$ is computed, then we can (i) delay the computation of each of the nodes $v_1, \ldots, v_{\ell'-1}$ to immediately before this step, in the same order as before, and (ii) consider the first time when one of the $u_i$ is deleted, and bring forward the deletion of the rest of the nodes $u_1, \ldots, u_{\ell}$ to immediately after this step. This can only increase the number of red pebbles available at any point, and since any external out-neighbors are dependent on the entire level, all other computation steps remain valid.

On the other hand, if one of the nodes $u_1, \ldots, u_{\ell}$ has already been deleted by the step when $v_{\ell'}$ is computed, then we move the entire sequence of computing $v_1$, deleting $u_1$, computing $v_2$, deleting $u_2$, etc., computing $v_{\ell'-1}$ and deleting $u_{\ell'-1}$ to immediately before the computation of $v_{\ell'}$, and the deletion of $u_{\ell'}, \ldots, u_{\ell}$ (if these nodes exist at all) to immediately after computing $v_{\ell'}$. The computation of any external out-neighbor of the levels again remains valid. For $u_1, \ldots, u_{\ell}$, the deletion of these nodes has only been delayed up to the time of computing $v_{\ell'}$, and after this point, no out-neighbor of the level was computed anyway, since we assumed that one of $u_1, \ldots, u_{\ell}$ has already been deleted by this point. For $v_1, \ldots, v_{\ell'}$, computing any out-neighbor of this level must happen after the computation of $v_{\ell'}$ anyway. Furthermore, the number of available pebbles in any step can never decrease, because for each $1 _{\!} \leq _{\!} i _{\!} \leq _{\!} \min (\ell, \ell')$, the computation of $v_i$ is earlier than the deletion of $u_i$, both in the original and the new pebbling sequence.

This implies that we can always reorganize any valid pebbling (without decreasing the number of available pebbles at any point) into a sequence where for each consecutive pair of levels, we either proceed from the earlier to the later level in a subsequent set of steps, or we first compute the later level (in subsequent steps), and then later delete the earlier level (in subsequent steps). This allows for a significantly simpler analysis, where we can consider the process of proceeding to the next level (with the preceding level either deleted or not) as single logical step that happens at a given point in the pebbling.

\subsection{Analysis of the construction}

We continue by defining the level sizes for each of the level towers in our DAG, i.e., assigning concrete values to our parameters. Assume that our input graph in the clique problem has $N$ nodes and $M$ edges.

Let us select $a=7$, $b_2=9$ and $c_2=3$ for our constant values. This implies $c_2 < a$, $b_2 > a$, and also ensures that $(a-c_2)>(b_2-a)$, which will slightly simplify the analysis. For convenience, we select $c_1=(N+M) \cdot b_2$ and $b_1 = 2 \cdot c_1$. Due to $q \leq N$ in the original clique problem, this ensures $c_1 > q \cdot b_2$, as well as $b_1 > c_1 + q \cdot (b_2-a) + 1$ and $b_1>(N+M) \cdot a + 1$, as required for the analysis. Finally, we set $r$ high enough to ensure that $r > 1 + b_1 - N \cdot b_2 - M \cdot a$ for the last levels of the main chain, e.g., $r=3 \cdot (N+M) \cdot b_2$. The size of the resulting construction is still in upper bounded by $O(1) \cdot r + N \cdot (b_1 + O(1)) + M \cdot (c_1 + O(1))$, and hence $n=O((N+M)^2)$ is still polynomial in $N$ and $M$.

The main tower will consist of $9$ levels, of the following size:
\begin{itemize}
 \item 1st and 2nd level: $r - 1 - (N+M) \cdot a$
 \item 3rd: $r - 1 - (N+M-q) \cdot a - b_1 - (q-1) \cdot b_2$
 \item 4th and 5th: $r - 1 - (N+M-q) \cdot a - q \cdot b_2$
 \item 6th: $r - 1 - (N+M-q -1) \cdot a - q \cdot b_2 - c_1$
 \item 7th: $r - 1 - (N+M-q -{q \choose 2}) \cdot a - q \cdot b_2 - {q \choose 2} \cdot c_2$
 \item 8th and 9th: $r - 1 - b_1 - N \cdot b_2 - M \cdot a$
\end{itemize}

As mentioned, when $v'$ is incident to $e'$ in $G'$, we draw edges from (each node of) the last level of the tower of $v'$ to (each node of) the third level of the tower of $e'$. Finally, we draw some further edges between the main tower and the gadgets representing $G'$:
\begin{itemize}
 \item from the 1st level of the main tower to the 1st level of every node and edge tower,
 \item from the 1st level of every node and edge tower to the 2nd level of the main tower,
 \item from the 4th level of the main tower to the 2nd level of every edge tower,
 \item from the 2nd level of every edge tower to the 5th level of the main tower,
 \item from the last level of every node and edge tower to the last level of the main tower.
\end{itemize}

The analysis of the initial part of the construction is relatively simple. In the beginning, our only option is to pebble the 1st level of the main tower. This can only be followed by the 1st level of every edge and node gadget. This then enables the 2nd level of the main tower; we can proceed to this, and delete the 1st level of the main tower. There are no alternatives up to this point: the 2nd level of the edge gadgets cannot yet be computed since they have predecessors in the 4th level of the main tower, and the 2nd level of node gadgets cannot be computed because the first two levels of the main tower only leave $(N+M) \cdot a +1$ pebbles available besides the main tower, and our choice of $b_1$ ensures $b_1 > (N+M) \cdot a + 1$.

After this, we can immediately proceed to the 3rd level of the main tower; since this is a smaller level, this gives us $(b_1-a)+(q-1) \cdot (b_2-a) + 1$ further pebbles to use. The next levels in the edge gadgets still have unpebbled predecessors, so this only leaves two options: we can either proceed to the next level in some of the node gadgets, or proceed to the 4th level in the main tower. Since $(b_1-a)$ is much larger than $(q-1) \cdot (b_2-a)$ (and proceeding to the 4th level of the main tower only reduces the number of available red pebbles), we can only have at most one of the node gadget 2nd-levels pebbled at once. That is, if we want to proceed in more than one of the node gadgets, we need to do this sequentially: proceed to the 2nd level in one gadget, then proceed to the 3rd level in the same node gadget, and only then proceed to the 2nd level in other node gadgets. As such, at any point before proceeding to the 4th level of the main tower, we will have (possibly) several node gadgets already at the 3rd level, and at most one node gadget at the 2nd level. Furthermore, with $(b_1-a)+(q-1) \cdot (b_2-a) + 1$ available pebbles, we can only move to the 3rd level of at most $q$ node gadgets: when proceeding in the $(q+1)$-th node gadget, this would already require $q \cdot (b_2-a)$ extra pebbles in the 3rd levels of the already proceeded node gadgets, and $(b_1-a)$ pebbles to proceed to the 2nd level of the $(q+1)$-th node gadget.

Finally, after proceeding to the 4th level of the main tower, we only have $(N+M-q) \cdot a + q \cdot b_2 + 1$ free red pebbles to use outside of the main tower, which is smaller than $(b_1-a)$; as such, all node gadgets have to be already at the 3rd level by this point, or still at the 1st level. As such, at the point when we proceed to the 4th level of the main tower, the state of the DAG is easy to characterize: we have at most $q$ node gadgets at the 3rd level, and all the remaining node gadgets (plus all edge gadgets) still at the 1st level.

The next steps are once again simple to analyze. The 4th level in the main tower allows us to proceed to the 2nd level of each edge gadget, which then allows us to proceed to the 5th level in the main tower. There are once again no alternatives to this, since we only have $(N+M-q) \cdot a + q \cdot b_2 + 1$ pebbles outside of the main tower, and at least $(N+M) \cdot a$ of these are used in the node and edge gadgets, leaving at most $q \cdot (b_2-a) + 1$ unused (depending on how many node gadgets have been proceeded to the 3rd level earlier). With both $(b_1-a)$ and $(c_1-a)$ chosen larger than $q \cdot (b_2-a)$, this implies we cannot proceed to the next level in any gadget (i.e., 2nd level in further node gadgets, or 3rd level in an edge gadget) at this point.

From here we can proceed to the 6th level of the main tower (this is only separate from the 5th level for ease of presentation, since they have different roles). This frees up $(c_1-a)$ new pebbles compared to the 5th level, which allows us to proceed any number of edge gadgets first to the 3rd, and then the 4th level, sequentially. Recall that the 4th level of edge gadgets is smaller than the 2nd, so we are indeed motivated to proceed to the 4th level of as many edge gadgets as possible; however, we can only do this for edge gadgets where both of the incident nodes' gadgets are already at the 3rd level, since these are predecessors of the 3rd level of each edge gadgets. Note that the number of free pebbles at this point is at most $(c_1-a) + q \cdot (b_2 - a) + 1$ (if all node gadgets are still at the 1st level), and $(b_1-a)$ is significantly larger than this (since $b_1-c_1=c_1>N \cdot b_2$), so it is not possible to proceed to the 2nd level in any node gadgets at this point.

At some point, we need to proceed to the 7th level of the main tower. This 7th level is the largest one, only allowing
\[ (N+M) \cdot a + q \cdot (b_2-a) - {q \choose 2} \cdot (a-c_2) + 1 \]
pebbles outside of the main tower. This quantity is again larger than $c_1$ by design, so by the time we reach this 7th level, all the edge gadgets must still be at the 2nd level, or already proceeded to the 4th level. However, even in such a setting, we need a very specific case to be able to proceed to this 7th level at all. In particular, let $q_0$ denote the number of node gadgets that are on the 3rd level already (recall that $q_0 \leq q$). Altogether, the node gadgets occupy $N \cdot a + q_0 \cdot (b_2-a)$ red pebbles at this point. The edge gadgets occupied $M \cdot a$ pebbles before the 6th level of the main chain; this is now reduced by $(a-c_2)$ times the number of edge gadgets that already proceeded to the 4th level. Note that with $q_0$ node gadgets on the 3rd level, we can have at most ${q_0 \choose 2}$ edge gadgets that could proceed to the 4th level.

Furthermore, recall that $(a-c_2)>(b_2-a)$. This implies that if $q_0 < q$, then we have
\[ q_0 \cdot (b_2-a) - {q_0 \choose 2} \cdot (a-c_2) >  q \cdot (b_2-a) - {q \choose 2} \cdot (a-c_2) \, ; \]
this follows by first replacing $(b_2-a)$ by $(a-c_2)$ on both sides (we can do this since $q_0<q$), then dividing all terms by $(a-c_2)$, and getting $q_0 - {q_0 \choose 2} > q - {q \choose 2}$, which holds for $q_0 < q$. That is, if $q_0 < q$, then even if we could proceed to the 4th level in as much as ${q_0 \choose 2}$ edge gadgets, the number of pebbles required outside of the main tower will still be larger than the number of available pebbles, and hence we cannot proceed to the 7th level of the main tower. Even if $q=q_0$, it is only possible to free up enough pebbles if the number of edge gadgets on the 4th level is exactly ${q \choose 2}$; in this case, the number of pebbles used outside of the main tower will be exactly $(N+M) \cdot a + q \cdot (b_2-a) - {q \choose 2} \cdot (a-c_2)$. However, this can only happen if there are indeed ${q \choose 2}$ distinct edges in the original graph that are induced by our chosen $q$ nodes, i.e., if the set of nodes form a clique.

After this point, it is straightforward to finish a pebbling strategy. We can proceed to the 8th level of the main tower immediately; this allows for having $b_1 + N \cdot b_2  + M \cdot a + 1$ pebbles besides the main tower. This allows us to proceed to the 2nd and then 3rd level of each of the node gadgets sequentially, and then we can proceed to the 3rd and then 4th level of each of the edge gadgets sequentially. Finally, this allows us to proceed to the 9th level of the main tower.

This shows that it is possible to pebble the DAG without I/O operations if and only if there exists a clique of size $q$ in the original input graph.

\subsection{Multiple copies}

Note that the reduction so far produces a DAG where it is NP-hard to decide whether the optimal pebbling cost is $0$ or at least $1$.

In order to show a stronger result, we first add a single source and sink node to the DAG for simplicity; the source node becomes an in-neighbor of all nodes in the first level of the main tower, and the sink node becomes an out-neighbor of all nodes in the last level of the main tower. So far this does not affect any of the properties of the construction.

Then we can take any parameter $T$, and consider $T$ subsequent copies of the construction, where we always merge the sink node of the $i$-th copy and the source node of the $(i+1)$-th copy into a single node. Once again, this does not affect the basic behavior of the construction: when the merged $i$-th sink / $(i+1)$-th source is pebbled, we can immediately delete all remaining red pebbles from the $i$-th (and earlier) copies, since these are not needed in the future anymore. As such, the copies of the construction are indeed independent, and the optimal strategy is to apply the same pebbling in each of these, one after another. In particular, this also implies that if there exists a clique of size $q$ in $G'$, then the optimal pebbling in the construction with $T$ copies still has a cost of $0$, but if such a clique does not exist, then each separate copy will incur a cost of at least $1$, and hence the optimal pebbling cost is altogether at least $T$.

Recall that the size of our original construction was $O( (N+M)^2) \leq O(N^4)$. As such, given any constant $\varepsilon>0$, we can select a large enough constant $t$ that ensures $\frac{4}{t+4} < \varepsilon$, and set $T=N^t$. With that, the size of the construction is $n=\Theta(N^{t+4})$, so still polynomial in $N$, and we get $T > n^{1-\varepsilon}$ asymptotically, since $1-\varepsilon < \frac{t}{t+4}$. As such, this provides us a variant of the construction where it is already NP-hard to decide whether the optimum cost is $0$, or at least $n^{1-\varepsilon}$. Hence this also shows the inapproximability to an additive $n^{1-\varepsilon}$ term for any constant $\varepsilon>0$.

\subsection{Other pebbling problem variants}

Recall that there are different variants of the pebbling problem, where instead of source nodes being freely computable, they always need to be loaded from slow memory, and/or where instead of having a pebble of any color on sink nodes, we are explicitly required to have a blue pebble on each sink node in the end. Both of these modifications add a fixed, unavoidable cost of $1$ to our construction in SPP. As such, in these problem variants, our construction shows a slightly different claim, namely that it is already NP-hard to decide whether the optimum is as low as $1$ or $2$, or at least $n^{1-\varepsilon}$; this results in an inapproximability to any multiplicative factor or additive term of $n^{1-\varepsilon}$.

Similarly, the proof easily carries over to the SPP variant where the I/O rules replace the original pebble: these rules still cannot be used throughout the construction, as they incur cost.

\subsection{Discussion of Corollary \ref{cor:surplus_cost}}

Note that given Theorem \ref{th:inapprox}, Corollary \ref{cor:surplus_cost} follows easily. Consider MPP for the simplest case of $k=1$ processor; in this case, the cost of the computation steps is at least $n$ (possibly more in case of recomputations), and for a pebbling of cost $C$, the surplus cost is $C-n$. If our construction can be pebbled in with a cost of $0$ in one-shot SPP, then the same pebbling has cost $n$ in MPP, and thus the surplus cost is $0$. On the other hand, if any pebbling requires an I/O step or a recomputation step, then the resulting cost is at least $n+g$ or $n+1$, respectively, corresponding to a surplus cost of $g$ or $1$. (In fact, due to the multiple copies of the construction, the surplus cost is again at least $n^{1-\varepsilon} _{\!} \cdot _{\!} g$ or $n^{1-\varepsilon}$, as before). This means that it is NP-hard to approximate the surplus cost in MPP to any finite multiplicative factor (or any additive $n^{1-\varepsilon}$ term).

\end{document}